\documentclass[sigconf,nonacm]{acmart}
\setkeys{acmart.cls}{balance=false}
\usepackage{booktabs,array}
\usepackage[ruled,noend]{algorithm2e}
\usepackage{tikz}
\usetikzlibrary{arrows.meta,positioning}
\usepackage{mathtools}

\AtBeginDocument{\theoremstyle{acmdefinition}\newtheorem{remark}[theorem]{Remark}}

\setcopyright{none}
\acmYear{2026}
\acmDOI{}
\acmISBN{}

\newcommand{\msg}[1]{\textsc{#1}}
\newcommand{\m}[2]{\langle\msg{#1},#2\rangle}

\title{Optimistic Reliable Broadcast in Three Steps}

\author{Nikita Polianskii}
\affiliation{\institution{IOTA Foundation}\country{}}
\email{nikita.polyansky@gmail.com}

\begin{abstract}
We study the latency tradeoffs of signature-free asynchronous reliable
broadcast with optimistic two-step delivery. We consider $n$ parties,
of which at most $f$ are Byzantine. An optimistic fault budget $b\le f$
guarantees two-step delivery with an honest broadcaster and at most $b$
faults, while correctness tolerates $f$ faults. We write $(d_o,d_s,d_t)$
for optimistic latency, honest-broadcaster latency with up to $f$ faults,
and totality delay---the elapsed time from first to last honest delivery.

For $3f+1\le n<4f$, our first construction improves Opt-RBC's
$(2,4,4)$ guarantees to $(2,3,4)$, preserving its maximal optimistic
budget $b_*=\lfloor n/2\rfloor-f$. Thus maximal optimism is compatible
with the optimal three-step good case.

We prove a complementary lower bound: for $3f+1\le n<4f$, every
deterministic reliable broadcast protocol in our asynchronous model
with two-step optimistic budget $b\ge n-3f$ has totality delay at least
three, even when the broadcaster is honest. Our second construction
achieves $(2,3,3)$, matching this lower bound at maximal optimism when
$f\ge2$ and $\lceil(10f-3)/3\rceil\le n<4f$. Both constructions use
quadratic messages and tolerate adaptive corruption. In the matching
range, maximal optimism therefore costs exactly one totality step
relative to Bracha's $(3,3,2)$, without increasing good-case latency.
\end{abstract}

\begin{document}
\maketitle
\providecommand{\manuscripttitlebreak}{}
\manuscripttitlebreak

\section{Introduction}
\label{sec:intro}

Reliable broadcast lets a designated broadcaster send a value to $n$
parties, up to $f$ of which may be Byzantine. Honest parties must agree
on the delivered value, must deliver an honest broadcaster's value, and
must all deliver if any one of them does. The last property, totality,
is what makes a locally delivered value safe to act on inside a larger
protocol. Its latency, the time by which the last honest party can lag
behind the first, bounds how long a DAG-based consensus protocol must
wait before a delivered block is usable at every
party~\cite{ShresthaEtAl2026OptimisticRBC}.

We measure latency in \emph{steps}, where one step is the actual
network delay bound between honest parties, and write $(d_o,d_s,d_t)$ for a
protocol's optimistic latency, good-case latency, and totality delay
(Section~\ref{sec:model}). Bracha's signature-free protocol is
$(3,3,2)$: it delivers an honest broadcaster's value in three steps,
and once an honest party delivers, every honest party delivers within
two further steps~\cite{Bracha1987,YuEtAl2026SimpleIT}. Two-step
delivery under all $f$ faults requires $n\ge4f$ without
signatures~\cite[Theorem~4]{AbrahamRenXiang2021Latency}. Shrestha et
al.\ therefore relax the requirement: two-step delivery is required only
in executions in which at most $b$ parties are actually Byzantine, for
an \emph{optimistic budget}
$b\le f$~\cite{ShresthaEtAl2026OptimisticRBC}. For $3f+1\le n<4f$ the
largest feasible budget is $b_*=\lfloor n/2\rfloor-f$. Their protocol,
Opt-RBC,\footnote{We use the corrected arXiv v4 of May 2026, which
adds an ack layer to fix a termination issue in the earlier version;
its guarantees supersede those of the CCS 2025 version.}
attains it, but at a price: it is $(2,4,4)$, so good-case
latency rises from three steps to four and totality delay from two to
four. The superseded v2 claimed
two-step optimism with a three-step good case~\cite{ShresthaEtAl2025OptimisticRBCSuperseded}.
Confirm/ack establishes those two latency guarantees while retaining the
ack layer of the corrected protocol. Opt-RBC's reported bad-case latency
five is measured from invocation, whereas totality four follows the
elapsed convention of Yu et al.~\cite{YuEtAl2026SimpleIT}.
We ask whether this price is necessary:
can maximal optimism coexist with Bracha's good-case latency, and with
his totality delay?

\paragraph{Results}
The confirm/ack protocol improves Opt-RBC's $(2,4,4)$ to $(2,3,4)$
at the same maximal budget $b_*$ throughout $3f+1\le n<4f$, with
$O(n^2)$ messages. Thus two-step optimism need not increase the optimal
three-step good case. Paired witnesses achieve $(2,3,3)$ over the
range below. For $f\ge2$ and $3f+1\le n<4f$, our lower
bound rules out totality below three whenever $b\ge n-3f$, in the
deterministic model of Section~\ref{sec:model}. At $n=3f+1$, even one-fault
two-step optimism therefore costs an extra totality step. Our quadratic-message $(2,3,3)$ protocol attains
three at maximal optimism whenever
$n\ge\lceil(10f-3)/3\rceil$, including the entire interval for
$2\le f\le6$. Throughout this matching range, maximal optimism costs
\emph{exactly one totality step and no good-case step} relative to
Bracha. Table~\ref{tab:comparison} separates the construction's budget
from the lower bound's premise.

\begin{table}[t]
\caption{Latency in steps for $3f+1\le n<4f$, with
$b_*=\lfloor n/2\rfloor-f$ and $s=n-3f-1$. All protocols are
signature-free with $O(n^2)$ messages; types count the proposal as well
as subsequent message types.}
\label{tab:comparison}
\small\setlength{\tabcolsep}{2.5pt}
\begin{tabular}{@{}llrrrr@{\hspace{10pt}}l@{}}
\toprule
Protocol & Budget $b$ & Types & $d_o$ & $d_s$ & $d_t$ & Ref.\\
\midrule
Bracha & -- & 3 & 3 & 3 & 2 & \cite{Bracha1987}\\
Opt-RBC & $b_*$ & 5 & 2 & 4 & 4 & \cite{ShresthaEtAl2026OptimisticRBC}\\
Bracha + fast echo & $\lfloor s/2\rfloor+1$ & 3 & 2 & 3 & 3 & App.~\ref{sec:fast-bracha}\\
Confirm/ack & $b_*$ & 6 & 2 & 3 & 4 & App.~\ref{sec:protocol}\\
Paired witness & $\min\{b_*,2s+3\}$ & 6 & 2 & 3 & 3 & Sec.~\ref{sec:paired}\\
\midrule
Lower bounds & & -- & $2^{a}$ & $3^{b}$ & $3^{c}$ & \\
\bottomrule
\end{tabular}
\par\smallskip
{\footnotesize
$^a$Any budget (Lemma~\ref{lem:two-step-minimum}).
$^b$Under all $f$ faults~\cite[Theorem~4]{AbrahamRenXiang2021Latency}.
$^c$Deterministic protocols with two-step budget $b\ge n-3f$
(Theorem~\ref{thm:totality-lower}), witnessed with an honest broadcaster.
The fast-echo baseline has Byzantine-broadcaster totality at most two
(Proposition~\ref{prop:fast-bracha}). Confirm/ack and paired witnesses are
our constructions; paired witnesses require $f\ge2$.\par}
\end{table}

\paragraph{Full-range construction}
Confirm/ack retains Opt-RBC's fast echo delivery and vote--ack--ready
recovery. It adds a confirm sent on a sufficiently large echo quorum,
delivery on a confirm quorum, and an ack trigger on $f+1$ confirms.
With an honest broadcaster, echoes arrive by step two and confirms by
step three. The ack layer reconciles recovery from the echo and confirm
delivery rules when the broadcaster equivocates. This gives $(2,3,4)$
at budget $b_*$ throughout the interval; Theorem~\ref{thm:main} and
Appendix~\ref{sec:baseline-results} give the
thresholds and proof. We next ask whether its totality delay can be reduced.

\begin{theorem}[Totality lower bound]
\label{thm:totality-lower}
Let $3f+1\le n<4f$ and
$n-3f\le b\le\lfloor n/2\rfloor-f$. In the
deterministic model specified in Section~\ref{sec:model}, every
reliable broadcast protocol with two-step honest-broadcaster delivery
under at most $b$ faults has totality delay at least three. This holds
already over executions with an honest broadcaster and at most $f$
faulty non-broadcasters.
For $n\le4f-2$, the same bound holds if only non-broadcasters are
required to deliver and the totality clock counts only their deliveries.
\end{theorem}

The theorem allows arbitrary messages, timers, and continued
communication after delivery, and assumes neither threshold rules, nor
bounded message complexity, nor three-step good-case latency. It does not
cover randomized protocols or protocols supplied with the actual network delay bound.
An early delivery within the source group imposes a totality deadline
while some recipients are still waiting for messages. When that group
contains an ordinary recipient, its delivery suffices; only the boundary
$n=4f-1$ needs the source-inclusive convention in this proof.

The paired-witness construction attains this bound over the range below. Write
$s=n-3f-1$ for the number of parties beyond minimum resilience.

\begin{theorem}[Paired-witness broadcast]
\label{thm:paired}
Let $f\ge2$, $n=3f+1+s$ with $s\ge0$, and
\begin{equation}
 1\le b\le\min\{f,\lfloor n/2\rfloor-f,2s+3\}.
 \label{eq:paired-range}
\end{equation}
There is a signature-free asynchronous reliable broadcast protocol
against adaptive corruption of at most $f$ parties with optimistic
latency two under at most $b$ faults, good-case latency three, and
totality delay at most three. At most $5(n-1)^2+2(n-1)$ messages are sent while their authors
are uncorrupted, including after delivery, and local work is polynomial. For $n\ge4f$ it supports $b=f$.
\end{theorem}

The full-fault specialization at $n\ge4f$ recovers the elapsed-totality
guarantee already provided by Abraham, Ren, and Xiang's
protocol~\cite[Theorem~3]{AbrahamRenXiang2021Latency}.

\begin{corollary}[Matching totality bounds]
\label{cor:matching-totality}
For $f\ge2$ and $3f+1\le n<4f$ with $b_*\le2s+3$, equivalently
$n\ge\lceil(10f-3)/3\rceil$, the smallest totality delay of a
deterministic protocol in the model of
Section~\ref{sec:model} with maximal two-step optimism is exactly three,
and a quadratic-message protocol attains it together with the optimal
good-case latency three. For every $f\ge2$, at $n=3f+1$ totality
three is optimal for $1\le b\le\min\{3,b_*\}$.
\end{corollary}

The proof is at the end of Section~\ref{sec:paired}. At minimum
resilience the bounds match for every supported budget, and reach
maximal optimism for $2\le f\le6$. The $b=1$ case, including $(n,f)=(7,2)$,
already follows from Bracha with fast echo delivery
(Proposition~\ref{prop:fast-bracha}); paired witnesses extend the
three-totality guarantee to budgets two and three at minimum resilience.
The first missing maximal-budget case is
$(n,f,b_*)=(22,7,4)$, whereas adding one party gives the matching case
$(23,7,4)$. At the missing case the current rules actually lose
totality; Proposition~\ref{prop:paired-boundary} gives the general
counterexecution beyond their budget range. In general the sufficient surplus is
$s\ge\max\{0,\lceil(f-6)/3\rceil\}$. The lower bound holds throughout
$3f+1\le n<4f$ at budgets $b\ge n-3f$; maximal optimism with three-step totality at minimum
resilience for arbitrary $f$ remains open. The bound concerns
honest-broadcaster delivery spread: it forces $d_t^{\mathrm H}=3$
when $d_s=3$, but does not rule out $d_t^{\mathrm B}=2$. At $n=4f-1$,
including $(n,f)=(7,2)$, the lower bound uses the source-inclusive
delivery convention.

\paragraph{Techniques}
The lower bound builds on the four-execution indistinguishability
arguments of Abraham, Ren, and
Xiang~\cite[Theorem~4]{AbrahamRenXiang2021Latency} and of Shrestha et
al.~\cite[Theorem~8]{ShresthaEtAl2026OptimisticRBC}.
The step from an early fast delivery to a deadline for other parties
also appears in Abraham, Ren, and Xiang's argument~\cite[Theorem~7]{AbrahamRenXiang2021Latency}: their first bridge
preserves a two-round delivery and invokes the assumed bad-case bound.
Our elapsed-time argument preserves an early delivery under nonuniform
larger delays, keeps the resulting deadline fixed across the required
history comparisons, and uses eventual totality alone in the final
bridges. A source group containing
ordinary recipients is honest in the fast and forcing executions but
faulty in the final bridges. It generalizes the padding set $E$, of size
at most one, in Shrestha et al.'s proof of Theorem~8, which is honest in
their two honest-broadcaster worlds and faulty together with the
broadcaster in the other two~\cite[Section~4]{ShresthaEtAl2026OptimisticRBC}.
This reduces the fast executions' fault
count to $n-3f$, giving the smaller-budget bound.

The protocol must reconcile recovery from an echo delivery with recovery
from a confirm delivery. Votes propagate the former, but using confirms
alone for the latter can produce conflicting ready messages. Our
\emph{witness} instead requires an echo and a confirm from each of the
same authors. These paired observations occupy immutable echo slots:
a witness quorum forces enough honest parties to be unable to vote for
an opposite value, even earlier in the execution. This exclusion lets
both recovery paths trigger ready directly, removing the acknowledgment
round while preserving agreement (Section~\ref{sec:paired}).

\paragraph{Consequences for BFT consensus}
In protocols that use reliable broadcast to disseminate blocks,
good-case latency controls how quickly honest proposals become available,
while totality controls how quickly other honest parties catch up once
one party delivers. When the same RBC-delivered evidence suffices for
round advancement at every party, this catch-up bound also limits
round-entry skew. Our improvements can therefore shorten both proposal
dissemination and the synchronization margins used for recovery.
Section~\ref{sec:applications} instantiates these consequences for
Sailfish++ and Simple-IT~\cite{ShresthaEtAl2026OptimisticRBC,YuEtAl2026SimpleIT}.

\paragraph{Supporting results}
For large values, a coded realization preserves the paired-witness
protocol's $(2,3,3)$ guarantees and supported optimistic budgets with
$O(nL+\kappa n^2\log n)$ bits in every execution, where $L$ is the
value length and $\kappa$ is the hash length. It uses collision-resistant
hashing as its only cryptographic assumption and retains security
against adaptive corruption (Section~\ref{sec:extensions}).
We also determine how the totality delay of our confirm/ack construction
depends on the actual number of faults at $n=3f+1$.
Its exact worst-case delay is one step
with no faults, two with one fault, three with two through $f-b_*+1$
faults, and four above that threshold, omitting empty ranges
(Proposition~\ref{prop:confirm-ack-profile}). Thus our construction needs
four-step totality only when more than roughly half the fault tolerance is used.

\paragraph{Organization}
Section~\ref{sec:model} fixes the model and the latency conventions.
Section~\ref{sec:totality-lower} proves the lower bound, with the
complete joint schedule in Appendix~\ref{app:totality-lower}.
Section~\ref{sec:paired} gives the paired-witness protocol and its
proofs. Section~\ref{sec:extensions} states the applications and
extensions, Section~\ref{sec:related} discusses related work, and
Section~\ref{sec:conclusion} lists open problems.

\section{Model and definitions}
\label{sec:model}

\paragraph{Parties and faults}
There are $n\ge3f+1$ parties, one of which is the designated
broadcaster. An adaptive adversary may corrupt parties during the
execution, at most $f$ in total; corruption is permanent, and a corrupted
party may deviate arbitrarily. A party is \emph{honest} if it is never
corrupted. Validity assumes an honest broadcaster; agreement, totality,
and all latency guarantees concern honest parties. If the broadcaster is
Byzantine, at most $f-1$ of the $n-1$ non-broadcasters are, and fixing
the eventual corrupted set gives the quorum arguments below.

\paragraph{Inputs and initialization}
The value domain $\mathcal V$ has at least two elements. Only the
broadcaster initially knows its input. All other parties have public,
input-independent initial states. Each instance starts without private
correlated setup or transferable authentication from earlier instances.

\paragraph{Channels}
Parties are connected by pairwise authenticated channels: a receiver knows
which party sent each message, and a Byzantine party cannot send messages in
the name of another party.  Every message sent between honest parties is
eventually received.  Channels need not preserve order.  Authentication is
not transferable: a party cannot prove to a third party what it received.
These channels are the only authentication resource; no signatures or
additional setup are provided.

\paragraph{Messages}
In our constructions, every message carries the instance identifier, a
message type, and a value from a domain $\mathcal V$ with a fixed encoding
of at most $B$ bits; a party ignores messages outside this format.  A
party sends each message type at most once per instance.  When it does, it
sends the message to all parties and counts its own message immediately if
it is itself an eligible sender of that type.  A receiver ignores a
message whose type its sender is not eligible to send.  For each message
type and each eligible author, a receiver keeps a \emph{slot}: a cell that
is initially empty, is filled by the first message of that type it
receives from that author, and is never changed afterwards; later
messages of the same type from the same author are ignored.  Consequently
every quorum below is a set of distinct parties with one message each.  A
Byzantine party may appear in quorums for different values at different
receivers, but cannot take the place of another party.

\paragraph{Processing}
Local processing is finite and instantaneous. Activity is locally finite:
a party has finitely many local events in any bounded time interval.
A party processes received messages one at a time and applies every rule
that becomes enabled.  Delivering a value does not stop a party from sending
its remaining messages.  Agreement, validity, and totality need only eventual
receipt of messages and fair processing. Our guarantees do not include
local termination at delivery: a party must retain the instance and
process events that can enable its remaining one-shot sends. After it
has delivered and sent every type for which it is eligible, it may
discard the instance. An unsent type need not ever become enabled, so
we give no bounded-time stopping guarantee after delivery.

\paragraph{Communication}
Communication bounds count every message or bit that any party
sends while it is uncorrupted, including messages sent after delivery and
in executions in which no party delivers.  Traffic sent after corruption
is not charged.

\begin{definition}[Reliable broadcast]
\label{def:rbc}
The broadcaster is invoked with a value $v\in\mathcal V$.  Each honest party
delivers at most one value.  A protocol is a reliable broadcast protocol if
every execution satisfies the following.
\begin{description}
\item[Validity.] If the broadcaster is honest, every honest party eventually
  delivers the broadcaster's value, and no honest party delivers another
  value.
\item[Agreement.] No two honest parties deliver different values.
\item[Totality.] If an honest party delivers a value, every honest party
  eventually delivers a value.
\end{description}
\end{definition}

\begin{definition}[Latency]
\label{def:latency}
For latency statements only, let $\delta>0$ be an
\emph{actual network delay bound}: every message between honest parties
in the execution is received within $\delta$ time. One \emph{step} is
$\delta$ time. The protocol is not given $\delta$; its code and initial
states do not depend on it. A bound of $d$ steps must hold for the same
protocol for every $\delta>0$ and every execution satisfying that
honest-link bound. We reserve $\Delta$ for the known upper bound that
holds after stabilization in the consensus applications. Correctness of
the broadcast protocols itself requires only eventual receipt.
The \emph{good-case latency} is the largest time from an honest broadcaster's
invocation until every honest party has delivered, over all executions with
at most $f$ Byzantine parties~\cite{AbrahamRenXiang2021Latency}.  The
\emph{optimistic good-case latency} with budget $b$ is the same quantity
over executions with at most $b$ Byzantine
parties~\cite{ShresthaEtAl2026OptimisticRBC}.  The \emph{totality delay} is
the supremum, over all executions with any broadcaster, of the time between
the first delivery by an honest party and the
last~\cite[Definition~11]{YuEtAl2026SimpleIT}. We write
$(d_o,d_s,d_t)$ for these bounds in steps, in the order optimistic,
good-case, and totality. A totality lower bound of three means that
every bound $D<3$ is violated by some execution; the supremum need
not be attained by one execution.
\end{definition}

Write $d_t^{\mathrm H}$ and $d_t^{\mathrm B}$ for the same totality
supremum restricted to honest and Byzantine broadcasters, respectively.
Then $d_t=\max\{d_t^{\mathrm H},d_t^{\mathrm B}\}$ and
$d_t^{\mathrm H}\le d_s$: with an honest broadcaster, every delivery
occurs between invocation and the good-case deadline. Our lower bound
concerns $d_t^{\mathrm H}$, and hence also $d_t$; it does not lower-bound
$d_t^{\mathrm B}$. For protocols with $d_s=3$ it makes this elementary
honest-broadcaster bound tight under the stated optimism premise.

\paragraph{Latency conventions}
Our triples always use elapsed totality as their third entry. Abraham,
Ren, and Xiang report three good-case rounds and four bad-case rounds
for Bracha, charging one extra round for
totality~\cite[Definition~5]{AbrahamRenXiang2021Latency}.  Shrestha et
al.\ report bad-case latency five measured from invocation, given that some honest
party delivered within the good-case
latency~\cite[Definition~6]{ShresthaEtAl2026OptimisticRBC}; with a
Byzantine broadcaster there may be no delivery at all, so we make no
claim about the time from invocation in that case.  Following Yu et
al.~\cite[Table~3]{YuEtAl2026SimpleIT} we report elapsed totality
delay: two for Bracha and four for Opt-RBC. In a nonuniform-delay
execution, one party can receive a ready quorum early while the honest
readies in that quorum take a full step to reach the others; their
newly triggered readies can take a second step. Thus an early first
delivery does not turn Bracha's elapsed totality bound into one step.
We use the elapsed bounds directly, without assuming a general
conversion from round-based bounds. Table~\ref{tab:comparison} uses
this convention.

\paragraph{Scope of the lower bound}
\label{par:lower-scope}
The lower bound concerns deterministic protocols in the model above,
with latency measured against the unknown actual network delay bound.
The default interface requires the honest broadcaster to deliver, and
its delivery can start the totality clock.
A party's \emph{complete timed history} is the time-stamped sequence of
its receipts, with sender identity and content, its timer events, state
transitions, sends, and deliveries; its restriction to $[0,\tau)$ is its
history \emph{before} $\tau$.
All parties start the instance at time $0$ and can read real time;
this only strengthens the protocols quantified over.
Arbitrary recipient-specific packets, repeated sends, reactive timers, and
continued service after delivery are allowed.  The one-shot conventions
above restrict our constructions, not the protocols quantified over in
the lower bound.  The lower bound makes no claim for randomized
protocols or for protocols supplied with the actual network delay bound.
Theorem~\ref{thm:totality-lower} additionally covers a recipient-only
interface when $n\le4f-2$: require validity, agreement,
totality, and optimistic delivery for all honest non-broadcasters, and
measure the first-to-last delivery gap only over them. The broadcaster
remains one of the $n$ communicating parties and counts towards $f$.

\section{The totality lower bound}
\label{sec:totality-lower}

We prove Theorem~\ref{thm:totality-lower} for arbitrary deterministic
protocols in the model of Section~\ref{sec:model}, and also prove the
stated extension in which only non-broadcasters must deliver. Denote the
broadcaster by $0$. All
fault sets in the proof are static, so the bound also applies to
protocols that tolerate adaptive corruption.

\paragraph{Proof plan}
The six executions serve three purposes. Fast executions force an honest
party to deliver early. Forcing executions preserve that early
delivery while delaying other parties, so a hypothetical bound $D<3$
imposes a deadline on those parties. Bridge executions inherit opposite
deliveries but give their common honest parties identical histories
forever, contradicting agreement through eventual totality. Only the
forcing executions use the numeric deadline; the bridges use no latency
bound. We now construct the histories that make these transfers possible.

\paragraph{Partition and executions}
Write $s=n-3f-1$, so $0\le s\le f-2$ and $s+1\le b$.
Partition all parties into a \emph{source group} $\mathcal S$ containing $0$,
two \emph{endpoint} groups $E_A,E_B$, and two \emph{common} groups
$K_A,K_B$ with
\[
 |\mathcal S|=f-s-1,\qquad |E_A|=|E_B|=s+1,\qquad |K_A|=|K_B|=f.
\]
All five groups are nonempty and their sizes sum to $n$.
Only $0$ initially knows the input; $\mathcal S\setminus\{0\}$ consists of
ordinary recipients with input-independent initial states.
Choose an \emph{anchor} $h\in \mathcal S\setminus\{0\}$ if $|\mathcal S|\ge2$;
otherwise choose $h=0$. Only this last case, $n=4f-1$, requires the broadcaster to deliver.
Fix distinct inputs $A,B$. Table~\ref{tab:lower-worlds} in
Appendix~\ref{app:totality-lower} lists six executions: two \emph{fast}
executions $F_A,F_B$ forcing early anchor delivery; two
honest-broadcaster \emph{forcing} executions $X_A,X_B$; and two
Byzantine-broadcaster \emph{bridges} $W_A,W_B$. Their faulty sets are
$E_B,E_A,K_B,K_A,\mathcal S\cup E_B,\mathcal S\cup E_A$, respectively. Thus each fast
execution has $s+1\le b$ faults, and each forcing execution and bridge
has exactly $f$. The extra members of $\mathcal S$ are honest in the fast and
forcing executions but Byzantine in both bridges.
Figure~\ref{fig:lower} shows $W_A$.

Three terms recur. The \emph{stream} of a party $q$ to a party $p$ in
an execution is the timed sequence of packets $q$ sends to $p$, with
their arrival times; a Byzantine party \emph{copies} a stream of another
execution when it sends the same packets under its own identity with the
same arrival times, which nontransferable authentication permits. An
\emph{initial packet} is one sent during zero-time initialization,
before any external receipt. Non-broadcasters' initial states, and hence
their initial packets, are input-independent. All scheduled link delays below concern interparty
packets; self-messages are local events. Complete timed histories are as defined in
Section~\ref{sec:model}.

\begin{figure}[t]
\centering
\begin{tikzpicture}[>=Latex,font=\fontsize{7.5}{9}\selectfont,
  grp/.style={draw,rounded corners=2pt,minimum width=8.5mm,minimum height=5mm,inner sep=1pt},
  byz/.style={grp,fill=black!15},
  lbl/.style={fill=white,inner sep=1pt}]
\node[byz] (z) at (3.0,1.55) {$\mathcal S\ni0$};
\node[above=0.5mm of z,text=black!70] {$f-s-1$};
\node[grp] (ea) at (0,0) {$E_A$};
\node[grp] (ka) at (2.0,0) {$K_A$};
\node[grp] (kb) at (4.0,0) {$K_B$};
\node[byz] (eb) at (6.0,0) {$E_B$};
\node[below=0.5mm of ea,text=black!70] {$s+1$};
\node[below=0.5mm of eb,text=black!70] {$s+1$};
\draw[->] (z) -- node[lbl,pos=.5]{copy of $X_A$} (ea);
\draw[->] (z) -- node[lbl,pos=.5]{$X_A$} (ka);
\draw[->] (z) -- node[lbl,pos=.5]{$X_B$} (kb);
\draw[->,dashed] (eb) to[bend left=22] node[lbl,pos=.4,below=4pt]{copy of $W_B$} (ka);
\draw[->,dashed] (eb) to[bend left=22] (kb);
\draw[->,dashed] (eb) to[bend left=38] node[lbl,pos=.5]{late $E_B$ of $X_A$} (ea);
\draw[->] (ea) to[bend right=38] node[lbl,pos=.5]{delay $J$} (kb);
\draw[->] (ea) to[bend right=20] node[lbl,pos=.5]{$\eta$} (ka);
\draw[<->] (ka) -- node[lbl]{$\eta$} (kb);
\node[below=0.5mm of ka,text=black!70,fill=white,inner sep=1pt] {$f$};
\node[below=0.5mm of kb,text=black!70,fill=white,inner sep=1pt] {$f$};
\end{tikzpicture}

\medskip
\begin{tikzpicture}[>=Latex,font=\fontsize{7.5}{9}\selectfont]
\draw[->] (0,0) -- (7.4,0);
\foreach \x/\l in {0.25/$\varepsilon$,0.5/$2\varepsilon$,1.75/$\eta$,3.5/$2\eta$,4.6/$H$,5.25/$3\eta$,7.0/$J=4\eta$}
  \draw (\x,0.07) -- (\x,-0.07) node[below]{\l};
\node[above,align=center] at (0.5,0.12) {anchor\\delivers};
\node[above,align=center] at (1.75,0.12) {late\\receipts};
\node[above,align=center] at (3.5,0.12) {common\\cut};
\node[above,align=center] at (4.6,0.12) {forced\\delivery};
\node[above,align=center] at (5.25,0.85) {endpoint cut};
\node[above,align=center] at (7.0,0.12) {long-link\\delay};
\end{tikzpicture}
\caption{Top: the bridge $W_A$; shaded parties are Byzantine. The
members of $\mathcal S$ copy their $X_A$ streams towards the $A$ side and their
$X_B$ streams towards $K_B$. Only $0\in \mathcal S$ starts with an input.
The Byzantine $E_B$ copies the stream of honest $E_B$ in $W_B$ towards
the common groups, so its stream to $K_A$ inherits delay $J$, and copies
the late group's stream of $X_A$ towards $E_A$ (dashed). Arrow delays
apply to noninitial interparty packets: honest $E_A$ sends to $K_B$
with delay $J$, and other honest links have delay $\eta$.
Initial non-broadcaster packets arrive at $\varepsilon$.
$W_B$ is the mirror image. Bottom:
the three time scales and the deadlines of the proof, with
$2\varepsilon<\eta$ and $H=2\varepsilon+D\eta<3\eta$.}
\Description{A diagram of five party groups with the source group and the
right endpoint group shaded as Byzantine, arrows labelled by the copied
streams and delays, and a timeline marking epsilon, two epsilon, eta,
two eta, H, three eta, and J equal to four eta.}
\label{fig:lower}
\end{figure}

\paragraph{Three scales, one elapsed deadline}
Suppose a constant $D<3$ bounds totality delay in every
honest-broadcaster execution with at most $f$ faulty non-broadcasters.
Choose $\varepsilon,\eta>0$ with
\[
 2\varepsilon<\eta,\qquad
 H:=2\varepsilon+D\eta<3\eta,\qquad J:=4\eta.
\]
The second condition holds once $\varepsilon<(3-D)\eta/2$.
The protocol's code and initial states are independent of the actual
network delay bound (Definition~\ref{def:latency}). Its initialization
is therefore fixed as we vary $\varepsilon$ and $\eta$. We may take
$\eta$ small enough that $3\eta$ precedes every positive expiration
of a timer set during initialization. There are finitely many such timers
by finite initialization and locally finite activity; if there are none,
this condition is vacuous. For this cut, the noninitial packets in copied endpoint streams must
reach the opposite common group no earlier than $2\eta$; the initial
packets are matched separately. Thus any $J\ge2\eta$ suffices;
$J=4\eta$ keeps the deadlines distinct.

In the fast executions every interparty message delay is $\varepsilon$, and the
Byzantine endpoint group sends only the initial packets its members would
send and is then silent. Two-step optimistic delivery makes the honest
anchor $h$ deliver $A$ in $F_A$ and $B$ in $F_B$ by $2\varepsilon$.

In $X_A$ the honest group $E_B$, the \emph{late group}, receives every
initial interparty packet at time $\eta$, including on links within $E_B$;
all other initial packets arrive at $\varepsilon$; all later honest
messages take $\eta$, except that messages sent to $h$ by
time $\varepsilon$ arrive with delay $\varepsilon$ (the \emph{anchor
exception}). The Byzantine $K_B$ replays its $F_A$ packets to
$h$ for sends through $\varepsilon$ and otherwise copies the
stream of honest $K_B$ in $W_A$. Through time $\varepsilon$ each party in
$\mathcal S\cup E_A\cup K_A$ has received only initial interparty packets,
exactly as in $F_A$, so
it sends the same packets, and the exception delivers them to
$h$ by $2\varepsilon$ as in $F_A$; any interparty packet sent after
$\varepsilon$ reaches $h$ after $2\varepsilon$ in both
executions, and the late group sends nothing noninitial before
$\eta>2\varepsilon$. This includes source $0$'s sends to $h$:
its input and its initial receipts match the seed. Hence $h$'s history through
$2\varepsilon$ is the same in $X_A$ as in $F_A$, and it delivers $A$ by
$2\varepsilon$ in $X_A$ as well. All honest delays in $X_A$ are at most
$\eta$, the totality clock has started no later than $h$'s delivery,
and validity fixes the value, so every member of $E_A$ delivers $A$ by
$2\varepsilon+D\eta=H$. Symmetrically, $X_B$ delays $E_A$, has $K_A$
Byzantine, and forces every member of $E_B$ to deliver $B$ by $H$.

In the bridges, initial non-broadcaster packets arrive at $\varepsilon$.
Noninitial honest interparty packets take $\eta$, except that the
honest endpoint group sends to the opposite common group with delay $J$. In
both bridges every Byzantine member of $\mathcal S$ copies its own $X_A$ stream
towards $E_A\cup K_A$ and its own $X_B$ stream towards $E_B\cup K_B$. The Byzantine
endpoint group copies the stream of its honest counterpart in the other
bridge towards $K_A\cup K_B$, and copies the late group's stream of the
forcing execution towards the opposite endpoint group
(Figure~\ref{fig:lower}). Its copied stream to the opposite common group
therefore inherits delay $J$ and is silent before $J$ apart from initial
packets, matching the late group's silence in the forcing execution. The
four executions $X_A,X_B,W_A,W_B$ are constructed jointly by increasing
time: every Byzantine packet copies an honest packet of another execution
that was sent at least $\varepsilon$ before its arrival, so at each
instant honest processing in all four executions is determined before
the copies are formed (Appendix~\ref{app:totality-lower}).

\begin{lemma}[Causal cut]
\label{lem:causal-cut}
Party by party, the complete timed histories satisfy
\begin{align*}
 X_A|_{K_A}&=W_A|_{K_A}\ \text{ before }2\eta,&
 X_A|_{E_A}&=W_A|_{E_A}\ \text{ before }3\eta,\\
 X_B|_{K_B}&=W_B|_{K_B}\ \text{ before }2\eta,&
 X_B|_{E_B}&=W_B|_{E_B}\ \text{ before }3\eta.
\end{align*}
\end{lemma}
\begin{proof}[Proof sketch]
Compare, sender by sender, the packets a member of $K_A$ receives before
$2\eta$ in $X_A$ and in $W_A$. From $E_B$: in $X_A$ the late group
receives nothing before $\eta$ and no initial timer fires before
$3\eta$, so its first noninitial packet reaches $K_A$ no earlier than
$2\eta$; in $W_A$ the copied stream is silent before $J$ apart from
initial packets, which agree. From $\mathcal S$ and from $K_B$: copied explicitly,
from $X_A$ into $W_A$ by the Byzantine members of $\mathcal S$ and from $W_A$ into
$X_A$ by the Byzantine $K_B$. From $E_A$ and from within $K_A$: a packet
arriving before $2\eta$ was sent before $\eta$, from a history
strictly earlier in the induction, which matches. Now compare the packets
a member of $E_A$ receives before $3\eta$. From $\mathcal S$, $E_B$, and $K_B$:
copied explicitly. From $K_A$: a packet arriving before $3\eta$ was
sent before $2\eta$, inside the interval just established. From within
$E_A$: sent from a strictly earlier matching history. Every nonlocal
dependency is on a strictly earlier event, so the simultaneous induction
over both cuts is well founded. The $B$ side is symmetric.
Appendix~\ref{app:totality-lower} carries out the induction, including
initialization traffic, reactive timers, and messages sent after
delivery.
\end{proof}

\begin{lemma}[Common indistinguishability]
\label{lem:common}
Every party in $K_A\cup K_B$ has the same complete timed history in $W_A$
and in $W_B$, for all time.
\end{lemma}
\begin{proof}[Proof sketch]
Sender by sender: each member of $\mathcal S$ sends the same streams to the
common groups in both bridges; each endpoint identity is honest in one bridge
and an exact copy of that honest stream in the other; and
common-to-common packets have identical schedules and, by induction on
time, identical contents. The private $E_A$--$E_B$ streams differ, but
each such channel has a Byzantine party at one end in each bridge, so no
honest channel obligation is violated (Appendix~\ref{app:totality-lower}).
\end{proof}

Since $H<3\eta$, Lemma~\ref{lem:causal-cut} transfers the forced
deliveries: $W_A$ contains an honest delivery of $A$ by $E_A$, and $W_B$
an honest delivery of $B$ by $E_B$. Honest delays in the bridges are at
most $J$, so both are fair asynchronous executions, in which agreement
and eventual totality hold although the broadcaster is Byzantine and the
hypothesis $D<3$ says nothing. Totality forces each party in
$K_A\cup K_B$ to deliver in both bridges, agreement forces $A$ in $W_A$
and $B$ in $W_B$, and Lemma~\ref{lem:common} makes this impossible for a
deterministic party. The numeric bound $D$ was used only in the
honest-broadcaster executions $X_A,X_B$. This proves
Theorem~\ref{thm:totality-lower}.

\paragraph{Scope}
At $n=3f+1$ each endpoint is a singleton and $|\mathcal S|=f-1$, so the proof
needs only one optimistic fault, for every $f\ge2$. At $n=4f$, $\mathcal S$
would be empty, although it must contain the broadcaster. More generally,
this construction needs $n\le3f+b$: its common groups have total size
at most $2f$; $\mathcal S\cup E_A$ is faulty in $W_B$, so
$|\mathcal S|+|E_A|\le f$; and $E_B$ is faulty in $F_A$, so $|E_B|\le b$. This is a limit of the partition, not an
impossibility of proving a stronger bound for $b<n-3f$.

\section{Three-step totality with paired witnesses}
\label{sec:paired}

Fix $f\ge2$, $n=3f+1+s$, and a budget $b$ in the
range~\eqref{eq:paired-range}, and let $N=n-1=3f+s$ be the number of
non-broadcasters. Only non-broadcasters send echoes, votes, confirms,
and witnesses, and receivers ignore such messages from the broadcaster;
every party, including the broadcaster, can send ready and deliver. When
the broadcaster is Byzantine we write $t\le f-1$ for the number of
Byzantine non-broadcasters, so that $t+1\le f$ parties are corrupted in
total. Algorithm~\ref{alg:paired} is the complete protocol and
Figure~\ref{fig:paired-flow} in Appendix~\ref{app:paired-details}
shows its message flow. The thresholds are
\begin{equation}
\begin{aligned}
 Q&=N-b=3f+s-b,& C&=N-f=2f+s,\\
 M&=Q-f+1=2f+s-b+1,& P&=C-f+1=f+s+1,\\
 R&=N-f+1=2f+s+1,& R_{\mathrm w}&=2f-1,
\end{aligned}
\label{eq:paired-thresholds}
\end{equation}
and Table~\ref{tab:paired-thresholds} in Appendix~\ref{app:paired-details}
records the role of each. They
count non-broadcasters, of whom at most $t$ are Byzantine. The ready
thresholds $f+1$ and $2f+1$ count all parties, of whom at most $f$ are
Byzantine: $f+1$ readies contain an honest one, and the $n-f\ge2f+1$
honest parties can always fill a ready delivery quorum.

\begin{algorithm}[t]
\DontPrintSemicolon
\caption{Paired-witness broadcast; thresholds in~\eqref{eq:paired-thresholds}.}
\label{alg:paired}
\textbf{Broadcaster} with input $v$: multicast $\m{propose}{v}$.\;
\medskip\textbf{Each non-broadcaster:}\;
\Indp
On $\m{propose}{v}$ from the broadcaster, multicast $\m{echo}{v}$.\;
On $M$ echoes of $v$, multicast $\m{vote}{v}$.\;
On $C$ echoes of $v$, and having itself echoed $v$ (the \emph{own-echo
guard}), multicast $\m{confirm}{v}$.\;
On an echo and a confirm of $v$ from each of the same $P$ authors,
multicast $\m{witness}{v}$.\;
\Indm
\medskip\textbf{Every party:}\;
\Indp
On $R$ votes, or $R_{\mathrm w}$ witnesses, or $f+1$ readies of $v$,
multicast $\m{ready}{v}$.\;
On $Q$ echoes, or $C$ confirms, or $2f+1$ readies of $v$, deliver $v$.\;
\Indm
\medskip All multicasts and delivery are one-shot across values.
Continue processing and sending after delivery.\;
\end{algorithm}

Slots (Section~\ref{sec:model}) are never overwritten: once an echo slot
at an honest party holds $v$, it never held and will never hold another
value. The witness rule counts \emph{paired authors}, identities whose
echo slot and confirm slot at the receiver both hold $v$; a claim by
another party to have seen such a pair does not count. The own-echo
guard applies only to confirms. It makes the honest confirmers behind a
confirm delivery paired authors at every honest party, which drives the
witness path; votes and witnesses carry no such guard, so a party that
echoed another value can still help recovery. A ready sent on $R$ votes
or on $R_{\mathrm w}$ witnesses is a ready \emph{origin}; one sent on $f+1$ readies
is a \emph{relay}. Deliveries on $Q$ echoes or on $C$ confirms are
\emph{delivery origins}; a delivery on $2f+1$ readies contains an honest
ready.

\paragraph{Design rationale}
After an echo delivery, $M$ honest echoes reach everyone and drive the
vote/ready recovery path. After a confirm delivery, $P$ honest confirms
reach everyone. Letting these $P$ confirms directly trigger ready is
unsafe: for admissible $b\ge s+2$, an $R$-vote quorum for one value can
coexist with $P$ confirms for another (Appendix~\ref{app:paired-details}).
Opt-RBC and our confirm/ack protocol reconcile the paths through an
acknowledgment round. Paired witnesses instead bind each confirm to its
author's echo; Lemma~\ref{lem:paired-compatibility} proves that a witness
quorum excludes an opposite vote quorum, allowing both paths to trigger
ready directly.

\subsection{Compatibility of the recovery paths}
Let the broadcaster be Byzantine, with $t\le f-1$ Byzantine
non-broadcasters. For adaptive faults, take the eventual corrupted set;
``honest'' below means outside that set. Two sets of $X$ and $Y$ distinct
non-broadcasters share an honest party whenever $X+Y-N>t$. The budget
range gives
\begin{align}
 2Q-N&=3f+s-2b>t,& 2C-N&=f+s>t,\nonumber\\
 2R-N&=f+s+2>t,& Q+C-N&=2f+s-b>t,\nonumber\\
 Q+M-N&=2f+s-2b+1>t,& C+R-N&=f+s+1>t,
 \label{eq:paired-intersections}
\end{align}
where $Q+M-N>t$ holds uniformly for $t\le f-1$ exactly when
$b\le\lfloor n/2\rfloor-f$; the others use $b\le f$. In particular, honest confirms have a unique value. One
single-party fact is also used: $C+M-N=f+s-b+1>0$, so a party whose $C$
echo slots hold one value never has $M$ slots for another and never
votes for it, at any time. The following lemma is the additional
compatibility that lets witnesses bypass the acknowledgment round.

\begin{lemma}[Witness/vote compatibility]
\label{lem:paired-compatibility}
With a Byzantine broadcaster, if $R_{\mathrm w}$ witness slots at some honest party
hold $v$, then at no time do $R$ vote slots at any honest party hold a
value $w\ne v$.
\end{lemma}
\begin{proof}
Let $g$ be the number of honest non-broadcasters that ever confirm $v$.
At least $R_{\mathrm w}-t\ge f$ of the witness authors are honest. The conclusions
below are time-independent, so the order of events plays no role.

If $g\ge f$, any $f$ honest confirmers of $v$ hold $C$ echo slots for
$v$ and, by the single-party fact, never vote for $w$; the witness
quorum is not needed in this case. Only non-broadcasters vote, so at
most $N-f=R-1$ distinct parties ever send a vote for $w$, and $R$ vote
slots for $w$ cannot fill.

Suppose $g\le f-1$ and fix an honest witness author $i$. Among its $P$
paired authors at most $g$ are honest, so at least $P-g\ge P-f+1=s+2$
are Byzantine; each of their echo slots at $i$ holds $v$ by the pairing
and is never overwritten. Since $P>t$, at least one paired author is
honest. That honest confirmer holds $C$ echo slots for $v$, at least
$C-t\ge C-f+1=f+s+1$ of them from honest parties. Being honest, those
parties echoed only $v$, so their echo slots at $i$ are empty or hold
$v$, whether or not their echoes have reached $i$; being honest, they
are distinct from the $P-g$ Byzantine paired authors. Hence at least
\[
 (C-t)+(P-g)\ \ge\ (f+s+1)+(s+2)\ =\ f+2s+3\ \ge\ f+b
\]
echo slots at $i$ never hold $w$, using $b\le2s+3$. At most
$N-(f+b)=M-1$ slots at $i$ can hold $w$ at any time, so $i$ never votes
for $w$. Applying this to the at least $f$ honest witness authors, again
at most $R-1$ parties ever vote for $w$.
\end{proof}

\begin{table}[t]
\caption{Compatibility of delivery origins ($Q$ echoes, $C$ confirms)
and ready origins ($R$ votes, $R_{\mathrm w}$ witnesses). Entries are intersection
sizes exceeding $t$, or Lemma~\ref{lem:paired-compatibility}; the
$C$-confirm/$R$-vote case also uses $C+M>N$.}
\label{tab:agreement}
\small\setlength{\tabcolsep}{3.5pt}
\begin{tabular}{@{}lllll@{}}
\toprule
 & $Q$-echo & $C$-confirm & $R$-vote & $R_{\mathrm w}$-witness\\
\midrule
$Q$-echo & $2Q{-}N$ & $Q{+}C{-}N$ & $Q{+}M{-}N$ & $Q{+}C{-}N$\\
$C$-confirm & & $2C{-}N$ & $C{+}R{-}N$ & $2C{-}N$\\
$R$-vote & & & $2R{-}N$ & Lemma~\ref{lem:paired-compatibility}\\
$R_{\mathrm w}$-witness & & & & $2C{-}N$\\
\bottomrule
\end{tabular}
\end{table}

\begin{lemma}[Agreement and validity]
\label{lem:paired-safety}
Algorithm~\ref{alg:paired} satisfies agreement and validity.
\end{lemma}
\begin{proof}
First let the broadcaster be Byzantine. Honest confirms have a unique
value, since the echo quorums of two confirmers share an honest party by
$2C-N>t$, and that party echoed once. Every honest witness has an honest
confirmer among its $P>t$ paired authors, so honest witnesses have that
same value. Table~\ref{tab:agreement} lists, for each pair of origins,
the fact that forbids different values; all of them are intersection
counts on immutable slots and do not depend on the order of events.

Readies first. Two $R$-vote quorums share an honest party by $2R-N>t$,
and that party voted once. An $R_{\mathrm w}$-witness quorum contains an honest
witness, and Lemma~\ref{lem:paired-compatibility} excludes an opposite
vote quorum. Following relays backwards, each honest ready reaches in
finitely many steps an honest origin ready for the same value, since a
relay on $f+1$ readies contains an honest, strictly earlier ready; two
honest readies for different values would give two honest origins for
different values, which is excluded. Hence all honest readies carry one
value.

Deliveries next. A delivery on $Q$ echoes excludes an opposite honest
vote by $Q+M-N>t$ and an opposite honest confirm by $Q+C-N>t$; hence it
agrees with every ready origin, through the honest voter or the honest
paired confirmer that the origin contains, and with every confirm
delivery; two echo deliveries share an honest party by $2Q-N>t$. A
delivery on $C$ confirms excludes an opposite $R$-vote quorum, since the
two author sets share an honest party by $C+R-N>t$ and an honest
confirmer cannot vote oppositely by $C+M>N$; it agrees with every
witness origin by confirm uniqueness. A delivery on $2f+1$ readies
contains an honest ready and reduces to the ready case. These cases
cover every pair of origins.

Now let the broadcaster be honest, with at most $f$ Byzantine
non-broadcasters. Every honest echo carries the broadcaster's input.
The thresholds $M,C,P$ exceed $f$, and $R_{\mathrm w}=2f-1>f$ for $f\ge2$, so every vote,
confirm, and witness quorum contains an honest author, whose message
carries the input by induction along the message chain; the same holds
for the delivery quorums $Q,C>f$ and for every ready origin. Relays and
ready deliveries contain an honest ready and reduce to an honest origin.
Hence every delivered value is the broadcaster's input.
\end{proof}

\subsection{Delivery and communication}
\begin{lemma}[Good case and totality]
\label{lem:paired-latency}
Algorithm~\ref{alg:paired} has the latency guarantees of
Theorem~\ref{thm:paired} and satisfies eventual totality asynchronously.
\end{lemma}
\begin{proof}
With an honest broadcaster, everyone receives the proposal by step one.
At most $b$ faults leave $Q$ honest echoes, giving delivery by step two;
at most $f$ leave $C$ honest echoes and then $C$ honest confirms, giving
delivery by step three. This also bounds honest-broadcaster totality.

For a Byzantine broadcaster, let the first honest delivery of $v$ be at
$T$. Echo delivery leaves $Q-t\ge M$ honest echoes already sent.
By $T+\delta$ everyone receives them and votes $v$: the intersection
$Q+M-N>t$ excludes any earlier opposite vote. At least $N-t\ge R$
honest votes arrive by $T+2\delta$, giving ready and delivery by
$T+3\delta$. Confirm delivery leaves $C-t\ge P$ honest authors that
sent both echo and confirm. Their paired messages reach everyone by
$T+\delta$, giving witness, ready, and delivery by $T+3\delta$.
Confirm uniqueness excludes an earlier opposite witness; there is no
own-echo guard on witnesses. Ready delivery leaves $f+1$ honest readies,
which reach everyone within one step, giving relay and delivery within
two. Safety excludes opposite readies in all cases. These finite chains
also prove eventual totality without a delay bound.
\end{proof}

\begin{proof}[Proof of Theorem~\ref{thm:paired}]
Combine Lemmas~\ref{lem:paired-safety} and~\ref{lem:paired-latency}.
Each non-broadcaster sends at most five multicasts and the broadcaster at
most two: $5(n-1)^2+2(n-1)$ messages, including failed executions and
service after delivery. Later corruptions do not increase pre-corruption
sends. Full values cost $O(n^2L)$ payload bits, in addition to message
headers; slot and matched-author counters
require polynomial work. For $n\ge4f$, $b=f$ satisfies
\eqref{eq:paired-range}, covering all honest-broadcaster executions.
\end{proof}

\begin{proof}[Proof of Corollary~\ref{cor:matching-totality}]
For $n<4f$, $n-3f\le b_*$. In the stated range
$b_*\le\min\{f,2s+3\}$, so Theorems~\ref{thm:paired}
and~\ref{thm:totality-lower} give matching totality bounds at $b=b_*$.
At $n=3f+1$, the lower bound applies to every $b\ge1$, and the
construction supports $b\le\min\{3,b_*\}$. Good-case latency three
is optimal below $4f$ by~\cite[Theorem~4]{AbrahamRenXiang2021Latency}.
\end{proof}

\section{Applications and further results}
\label{sec:extensions}

\subsection{Applications to BFT consensus}
\label{sec:applications}
BFT protocols that use reliable broadcast as a black box for block
dissemination can use paired witnesses when their RBC interface and
fault model match. This includes DAG-based protocols whose vertices
are disseminated by RBC, as in Bullshark's abstract
construction~\cite[Sections~3--4]{SpiegelmanEtAl2022Bullshark}.
Bullshark and Narwhal/Tusk~\cite{DanezisEtAl2022NarwhalTusk}
belong to the certified-DAG family surveyed by
Raikwar et al.~\cite[Section~III-D]{RaikwarEtAl2024SoKDAG}.
Replacing Opt-RBC's $(2,4,4)$ by $(2,3,3)$
reduces good-case latency and totality delay by one step each, with
the optimistic budget of Theorem~\ref{thm:paired}. At $n=3f+1$ with
$f\ge7$ this budget is three, below $b_*$: with between four and $b_*$
actual faults, Opt-RBC still guarantees two-step delivery and paired
witnesses guarantee three steps. Confirm/ack keeps the budget $b_*$.
Dissemination on the commit path can therefore finish earlier, and
recovery timers based on totality can be shorter; the resulting consensus
bounds depend on the surrounding protocol. An interface that additionally
requires transferable certificates needs a separate compatibility argument.

For concrete consequences, consider Sailfish++ and Simple-IT.
Let $\delta$ be the actual network delay bound after GST and $\Delta>\delta$
its known bound. On honest-leader commit paths with resolved dependencies,
replacing Opt-RBC by paired witnesses reduces commit latency from
$5\delta$ to $4\delta$ in both the DAG-based Sailfish++
protocol~\cite[Section~7]{ShresthaEtAl2026OptimisticRBC} and
base Simple-IT~\cite[Lemmas~6--8]{YuEtAl2026SimpleIT}.
Totality also bounds round-entry skew. In Simple-IT, if an honest party
enters round $r$ at time $t\ge\mathrm{GST}$, all honest parties have
entered it by $t+d_t\delta$~\cite[Lemma~12]{YuEtAl2026SimpleIT}.
Substituting paired witnesses reduces this bound from $4\delta$ to
$3\delta$. Simple-IT sets its round timer to
$(d_s+d_t)\Delta$ and bounds eventual worst-case view duration by
$(d_s+d_t)\Delta+d_t\delta$, including views with faulty leaders.
Substituting the paired-witness triple $(2,3,3)$ reduces these bounds from $8\Delta$ to
$6\Delta$ and from $8\Delta+4\delta$ to $6\Delta+3\delta$,
respectively~\cite[Lemma~8 and Corollary~5]{YuEtAl2026SimpleIT}.
These are consequences of the source analyses under their stated
conditions. Appendix~\ref{app:consensus-applications} gives the complete
comparison, the optimistic fault budgets and readiness premises, and the
full-proposal requirement for Sailfish++'s first-message commit rule.

\subsection{Further results}
We give a coded realization of paired witnesses and a full-value fallback
that retains maximal optimism throughout the resilience range.

\paragraph{Linear-communication paired witnesses}
\begin{proposition}[Coded paired witnesses]
\label{prop:paired-coded}
For the parameters of Theorem~\ref{thm:paired}, an MDS code and a
collision-resistant hash function give a $(2,3,3)$ reliable broadcast
protocol for $L$-bit values, conditional on commitment binding.
Against a computationally bounded adaptive adversary, binding fails only
with negligible probability. Communication sent while its authors are
uncorrupted is $O(nL+\kappa n^2\log n)$ bits in every execution,
where $\kappa$ is the hash length. Its payload traffic has
leading bounds $3nL$ with an honest broadcaster and $6nL$ in every
execution.
\end{proposition}
Appendix~\ref{app:paired-coded} gives the packet rules, the exact
bounds, and the proof. Each recipient echoes its own indexed fragment; a
voter reconstructs and supplies its own fragment if it has not already
sent it; a confirm requires both reconstruction and the own-echo guard.
Every possible first delivery therefore leaves enough honest fragments
in transit to preserve the control protocol's deadlines. Since Merkle
membership alone does not certify a codeword, every reconstruction
re-encodes the vector and checks the committed root. The bound charges
pre-corruption sends and needs no static-fault assumption.

\paragraph{Maximal optimism for every $n$, at four totality steps}
\begin{theorem}[Confirm/ack broadcast]
\label{thm:main}
Let $f\ge1$, $n\ge3f+1$, and
$b=\min\{f,\lfloor n/2\rfloor-f\}$. The confirm/ack protocol is a
signature-free asynchronous reliable broadcast protocol against adaptive
corruption of at most $f$ parties. It guarantees $(2,3,4)$ at budget $b$,
and $(2,2,4)$ when $n\ge4f$, with $O(n^2)$ messages sent while their
authors are uncorrupted, including after delivery.
\end{theorem}
The protocol and proof are in Appendix~\ref{sec:baseline-results}.
At $n=3f+1$, Proposition~\ref{prop:confirm-ack-profile} gives its exact
fault-dependent totality profile. The matching executions concern this
protocol; four steps are not a lower bound for reliable broadcast in general.

\section{Related work}
\label{sec:related}

\paragraph{Optimistic broadcast and latency}
Shrestha et al.\ study optimistic reliable broadcast and use Opt-RBC in
Sailfish++~\cite{ShresthaEtAl2026OptimisticRBC}. We cite the corrected
arXiv v4, not the superseded CCS 2025 protocol; theorem and figure
numbers refer to v4. At
$n=3f+1$ we retain its fast echo threshold and its vote mechanism;
paired witnesses replace its acknowledgment round. Its Theorem~8 bounds
the optimistic budget by $\lfloor n/2\rfloor-f$ below $4f$, and
Theorem~4 of Abraham, Ren, and Xiang (Theorem~10 at OPODIS 2021)
requires $n\ge4f$ for two-step delivery under all $f$
faults~\cite{AbrahamRenXiang2021Latency}; Theorem~\ref{thm:totality-lower}
adds a totality bound already at budget $n-3f$, including budget one
at minimum resilience. Both cited lower bounds use
four-execution indistinguishability arguments. Our six-execution proof
preserves an early honest-party delivery under nonuniform delays to force
a relative totality deadline; its final bridges use eventual totality.
Yu et al.\ define elapsed totality
delay and tabulate it for existing protocols~\cite{YuEtAl2026SimpleIT}.

Abraham, Ren, and Xiang also study the cost of two-round delivery at
$n\ge4f$~\cite[Section~4]{AbrahamRenXiang2021Latency}. Their Theorem~5
rules out a two-round bad case for $f\ge2$; Theorem~7 gives a four-round
bad-case lower bound when $f\ge3$ and $n\le5f-2$, in their round
convention. Their Theorem~3 supplies a two-round good-case protocol at
$n\ge4f$ whose rules give elapsed totality at most three: after fast
delivery, honest acks reach everyone, followed by vote-1 and vote-2;
vote-2 delivery needs at most two further steps. This is a $(2,2,3)$
upper bound at full budget in our convention. Thus our
$b=f$ specialization at $n\ge4f$ recovers an existing guarantee.
Their Theorems~8 and~9 give the special $f=2$ construction and a
three-round bad case at $n\ge5f-1$. These are the direct precedents for
studying the totality cost of a fast path. We do not convert their lower
bounds into elapsed-time bounds: our theorem concerns restricted
optimistic budgets below $4f$ and is witnessed by honest-broadcaster
spread, not by a lower bound on Byzantine-broadcaster recovery.

Imbs and Raynal give a signature-free two-step protocol for $n>5f$;
the corrected archive states this resilience restriction
explicitly~\cite{ImbsRaynal2016TwoStep}; ARX's $n\ge5f-1$ construction
improves that resilience for its fast bad-case guarantee.
Abraham, Nayak, Ren, and
Xiang's good-case categorization treats authenticated broadcast
with signatures~\cite{AbrahamNayakRenXiang2021GoodCase}; it is distinct
from the signature-free lower bounds used here.

\paragraph{Communication and dispersal}
Opt-RBC v4 already includes a balanced coded realization with
$O(nL+\kappa n^2\log n)$ bits~\cite[Section~5]{ShresthaEtAl2026OptimisticRBC}.
DXR~\cite{DasXiangRen2021ADD} and
CC-BRB~\cite{AlhaddadEtAl2022HashBRB} use hashing without trusted setup
and have $O(nL+\kappa n^2)$ bits; CC-BRB also has a three-step
honest-broadcaster path. Appendix~\ref{app:paired-coded} compares their
payload and metadata constants with ours.
Unsigned OptMiniCast has leading payload $1.5nL$ with an honest
broadcaster and $2nL$ universally~\cite{LocherShoup2024MiniCast}, and
MiniCast* reaches $1.5nL$ universally with hash functions; its two-step
variants use signatures~\cite{LocherShoup2025MiniCastStar}. Both beat
the universal payload coefficient of our coded realization, which instead
combines a two-step optimistic path with three-step totality.
Our coefficients are upper bounds, not lower bounds.
Merkle-authenticated fragments and reconstruction by re-encoding are established in
dispersal~\cite{CachinTessaro2005AVID,AlhaddadReyzinVaria2024CommittingAVID,
LocherPignolet2026AVID,AlhaddadEtAl2022HashBRB}. Our coded proof shows
availability at each delivery deadline of the paired-witness protocol.

\paragraph{Other termination models}
Lewis-Pye and O'Grady study unconditional Byzantine broadcast termination
under partial synchrony~\cite{LewisPyeOGrady2026Carnot}; those bounds do
not transfer directly to asynchronous reliable broadcast. Availability-graded
broadcast omits totality and has a separate resolver~\cite{Polyanskii2026Vantage}.
Quit-resistant broadcast permits stopping behaviors~\cite{MizrahiErbesWattenhofer2024Gather}
that our constructions do not support: parties complete their bounded
sends after delivery.

\section{Conclusion and open problems}
\label{sec:conclusion}

For $3f+1\le n<4f$, confirm/ack achieves $(2,3,4)$ at maximal optimism,
preserving the optimal three-step good case. Our lower bound shows that
totality delay at least three is unavoidable for deterministic protocols
in our model whenever $b\ge n-3f$, even with an honest broadcaster.
Paired witnesses attain this bound in the matching parameter ranges
established above, including maximal optimism. Both constructions use
quadratic messages.

Two directions remain. First, can a quadratic-message $(2,3,3)$ protocol
achieve maximal optimism throughout $3f+1\le n<4f$? The principal gap is
at minimum resilience, $n=3f+1$, where paired witnesses support only
$b\le\min\{3,b_*\}$. Second, what is the optimal totality delay for
smaller optimistic budgets $1\le b<n-3f$, which our lower bound does
not cover? In particular, which parameters permit $(2,3,2)$?

\bibliographystyle{ACM-Reference-Format}
\bibliography{references}

%
%
%

\begin{thebibliography}{20}

%
%
%
%
%
%
%
%
%
%
%
%
%
%
%
%

\ifx \showCODEN    \undefined \def \showCODEN     #1{\unskip}     \fi
\ifx \showDOI      \undefined \def \showDOI       #1{#1}\fi
\ifx \showISBNx    \undefined \def \showISBNx     #1{\unskip}     \fi
\ifx \showISBNxiii \undefined \def \showISBNxiii  #1{\unskip}     \fi
\ifx \showISSN     \undefined \def \showISSN      #1{\unskip}     \fi
\ifx \showLCCN     \undefined \def \showLCCN      #1{\unskip}     \fi
\ifx \shownote     \undefined \def \shownote      #1{#1}          \fi
\ifx \showarticletitle \undefined \def \showarticletitle #1{#1}   \fi
\ifx \showURL      \undefined \def \showURL       {\relax}        \fi
%
%
\providecommand\bibfield[2]{#2}
\providecommand\bibinfo[2]{#2}
\providecommand\natexlab[1]{#1}
\providecommand\showeprint[2][]{arXiv:#2}

\bibitem[Abraham et~al\mbox{.}(2021a)]%
        {AbrahamNayakRenXiang2021GoodCase}
\bibfield{author}{\bibinfo{person}{Ittai Abraham}, \bibinfo{person}{Kartik
  Nayak}, \bibinfo{person}{Ling Ren}, {and} \bibinfo{person}{Zhuolun Xiang}.}
  \bibinfo{year}{2021}\natexlab{a}.
\newblock \showarticletitle{Good-case Latency of {Byzantine} Broadcast: A
  Complete Categorization}. In \bibinfo{booktitle}{\emph{Proceedings of the ACM
  Symposium on Principles of Distributed Computing (PODC)}}.
  \bibinfo{pages}{331--341}.
\newblock
\urldef\tempurl%
\url{https://arxiv.org/abs/2102.07240v1}
\showURL{%
\tempurl}
\newblock
\shownote{Cited archive: arXiv:2102.07240v1, February 14, 2021}.


\bibitem[Abraham et~al\mbox{.}(2021b)]%
        {AbrahamRenXiang2021Latency}
\bibfield{author}{\bibinfo{person}{Ittai Abraham}, \bibinfo{person}{Ling Ren},
  {and} \bibinfo{person}{Zhuolun Xiang}.} \bibinfo{year}{2021}\natexlab{b}.
\newblock \showarticletitle{{Good-Case and Bad-Case Latency of Unauthenticated
  Byzantine Broadcast: A Complete Categorization}}. In
  \bibinfo{booktitle}{\emph{25th International Conference on Principles of
  Distributed Systems (OPODIS 2021)}} \emph{(\bibinfo{series}{Leibniz
  International Proceedings in Informatics (LIPIcs)},
  Vol.~\bibinfo{volume}{217})}. \bibinfo{publisher}{Schloss Dagstuhl --
  Leibniz-Zentrum f{\"u}r Informatik}, \bibinfo{address}{Dagstuhl, Germany},
  \bibinfo{pages}{5:1--5:20}.
\newblock
\urldef\tempurl%
\url{https://doi.org/10.4230/LIPIcs.OPODIS.2021.5}
\showDOI{\tempurl}
\showeprint[arxiv]{2109.12454v1}
\newblock
\shownote{We cite arXiv:2109.12454v1 (2021-09-25); its Theorem 4 and Definition
  5 are Theorem 10 and Definition 7 in the OPODIS version}.


\bibitem[Alhaddad et~al\mbox{.}(2022)]%
        {AlhaddadEtAl2022HashBRB}
\bibfield{author}{\bibinfo{person}{Nicolas Alhaddad}, \bibinfo{person}{Sisi
  Duan}, \bibinfo{person}{Mayank Varia}, {and} \bibinfo{person}{Haibin Zhang}.}
  \bibinfo{year}{2022}\natexlab{}.
\newblock \bibinfo{title}{Practical and Improved Byzantine Reliable Broadcast
  and Asynchronous Verifiable Information Dispersal from Hash Functions}.
\newblock \bibinfo{howpublished}{ePrint 2022/171}.
\newblock
\urldef\tempurl%
\url{https://eprint.iacr.org/2022/171}
\showURL{%
\tempurl}


\bibitem[Alhaddad et~al\mbox{.}(2024)]%
        {AlhaddadReyzinVaria2024CommittingAVID}
\bibfield{author}{\bibinfo{person}{Nicolas Alhaddad}, \bibinfo{person}{Leonid
  Reyzin}, {and} \bibinfo{person}{Mayank Varia}.}
  \bibinfo{year}{2024}\natexlab{}.
\newblock \bibinfo{title}{{Committing AVID with Partial Retrieval and Optimal
  Storage}}.
\newblock \bibinfo{howpublished}{Cryptology ePrint Archive, Paper 2024/685}.
\newblock
\urldef\tempurl%
\url{https://eprint.iacr.org/2024/685}
\showURL{%
\tempurl}
\newblock
\shownote{May 4, 2024 submission}.


\bibitem[Bracha(1987)]%
        {Bracha1987}
\bibfield{author}{\bibinfo{person}{Gabriel Bracha}.}
  \bibinfo{year}{1987}\natexlab{}.
\newblock \showarticletitle{Asynchronous {Byzantine} Agreement Protocols}.
\newblock \bibinfo{journal}{\emph{Information and Computation}}
  \bibinfo{volume}{75}, \bibinfo{number}{2} (\bibinfo{year}{1987}),
  \bibinfo{pages}{130--143}.
\newblock
\urldef\tempurl%
\url{https://doi.org/10.1016/0890-5401(87)90054-X}
\showDOI{\tempurl}


\bibitem[Cachin and Tessaro(2005)]%
        {CachinTessaro2005AVID}
\bibfield{author}{\bibinfo{person}{Christian Cachin} {and}
  \bibinfo{person}{Stefano Tessaro}.} \bibinfo{year}{2005}\natexlab{}.
\newblock \showarticletitle{Asynchronous Verifiable Information Dispersal}. In
  \bibinfo{booktitle}{\emph{Proceedings of the IEEE Symposium on Reliable
  Distributed Systems (SRDS)}}. \bibinfo{pages}{191--202}.
\newblock
\urldef\tempurl%
\url{https://people.csail.mit.edu/tessaro/papers/dds.pdf}
\showURL{%
\tempurl}
\newblock
\shownote{Cited archive: author-hosted full version dated November 15, 2004}.


\bibitem[Danezis et~al\mbox{.}(2022)]%
        {DanezisEtAl2022NarwhalTusk}
\bibfield{author}{\bibinfo{person}{George Danezis}, \bibinfo{person}{Lefteris
  Kokoris-Kogias}, \bibinfo{person}{Alberto Sonnino}, {and}
  \bibinfo{person}{Alexander Spiegelman}.} \bibinfo{year}{2022}\natexlab{}.
\newblock \showarticletitle{Narwhal and Tusk: A {DAG}-based Mempool and
  Efficient {BFT} Consensus}. In \bibinfo{booktitle}{\emph{Proceedings of the
  Seventeenth European Conference on Computer Systems (EuroSys)}}.
  \bibinfo{pages}{34--50}.
\newblock
\urldef\tempurl%
\url{https://doi.org/10.1145/3492321.3519594}
\showDOI{\tempurl}


\bibitem[Das et~al\mbox{.}(2021)]%
        {DasXiangRen2021ADD}
\bibfield{author}{\bibinfo{person}{Sourav Das}, \bibinfo{person}{Zhuolun
  Xiang}, {and} \bibinfo{person}{Ling Ren}.} \bibinfo{year}{2021}\natexlab{}.
\newblock \showarticletitle{Asynchronous Data Dissemination and its
  Applications}. In \bibinfo{booktitle}{\emph{Proceedings of the ACM SIGSAC
  Conference on Computer and Communications Security (CCS)}}.
  \bibinfo{pages}{2705--2721}.
\newblock
\urldef\tempurl%
\url{https://eprint.iacr.org/2021/777}
\showURL{%
\tempurl}
\newblock
\shownote{Cited archive: ePrint 2021/777}.


\bibitem[Imbs and Raynal(2016)]%
        {ImbsRaynal2016TwoStep}
\bibfield{author}{\bibinfo{person}{Damien Imbs} {and} \bibinfo{person}{Michel
  Raynal}.} \bibinfo{year}{2016}\natexlab{}.
\newblock \showarticletitle{Trading off $t$-Resilience for Efficiency in
  Asynchronous {Byzantine} Reliable Broadcast}.
\newblock \bibinfo{journal}{\emph{Parallel Processing Letters}}
  \bibinfo{volume}{26}, \bibinfo{number}{4} (\bibinfo{year}{2016}),
  \bibinfo{pages}{1650017}.
\newblock
\urldef\tempurl%
\url{https://doi.org/10.1142/S0129626416500171}
\showDOI{\tempurl}
\newblock
\shownote{Cited archive: corrected arXiv:1510.06882v2, September 28, 2021;
  requires n > 5f}.


\bibitem[Lewis-Pye and O'Grady(2026)]%
        {LewisPyeOGrady2026Carnot}
\bibfield{author}{\bibinfo{person}{Andrew Lewis-Pye} {and}
  \bibinfo{person}{Patrick O'Grady}.} \bibinfo{year}{2026}\natexlab{}.
\newblock \bibinfo{title}{{The Carnot Bound: Limits and Possibilities for
  Bandwidth-Efficient Consensus}}.
\newblock
\newblock
\showeprint[arxiv]{2603.11797v3}
\urldef\tempurl%
\url{https://arxiv.org/abs/2603.11797v3}
\showURL{%
\tempurl}
\newblock
\shownote{arXiv:2603.11797v3; 2026-07-28}.


\bibitem[Locher and Pignolet(2026)]%
        {LocherPignolet2026AVID}
\bibfield{author}{\bibinfo{person}{Thomas Locher} {and}
  \bibinfo{person}{Yvonne-Anne Pignolet}.} \bibinfo{year}{2026}\natexlab{}.
\newblock \showarticletitle{{Asynchronous Verifiable Information Dispersal with
  Low Space and Communication Complexity}}. In
  \bibinfo{booktitle}{\emph{Proceedings of the 38th ACM Symposium on
  Parallelism in Algorithms and Architectures}}. \bibinfo{publisher}{ACM},
  \bibinfo{address}{New York, NY, USA}, \bibinfo{pages}{143--154}.
\newblock
\urldef\tempurl%
\url{https://doi.org/10.1145/3816782.3819202}
\showDOI{\tempurl}
\newblock
\shownote{arXiv:2608.24636v2, August 26, 2026}.


\bibitem[Locher and Shoup(2025)]%
        {LocherShoup2024MiniCast}
\bibfield{author}{\bibinfo{person}{Thomas Locher} {and} \bibinfo{person}{Victor
  Shoup}.} \bibinfo{year}{2025}\natexlab{}.
\newblock \showarticletitle{{MiniCast: Minimizing the Communication Complexity
  of Reliable Broadcast}}. In \bibinfo{booktitle}{\emph{Advances in Cryptology
  -- EUROCRYPT 2025, Part V}} \emph{(\bibinfo{series}{Lecture Notes in Computer
  Science})}. \bibinfo{publisher}{Springer}, \bibinfo{address}{Cham},
  \bibinfo{pages}{96--115}.
\newblock
\urldef\tempurl%
\url{https://doi.org/10.1007/978-3-031-91092-0_4}
\showDOI{\tempurl}
\newblock
\shownote{We cite ePrint 2024/571 (last revised 2024-04-26); section and figure
  numbers refer to that version}.


\bibitem[Locher and Shoup(2026)]%
        {LocherShoup2025MiniCastStar}
\bibfield{author}{\bibinfo{person}{Thomas Locher} {and} \bibinfo{person}{Victor
  Shoup}.} \bibinfo{year}{2026}\natexlab{}.
\newblock \showarticletitle{{Towards Reliable Broadcast with Optimal
  Communication and Round Complexity}}. In
  \bibinfo{booktitle}{\emph{Proceedings of the 38th ACM Symposium on
  Parallelism in Algorithms and Architectures (SPAA '26)}}.
  \bibinfo{publisher}{ACM}, \bibinfo{address}{New York, NY, USA},
  \bibinfo{pages}{155--166}.
\newblock
\urldef\tempurl%
\url{https://doi.org/10.1145/3816782.3819206}
\showDOI{\tempurl}
\newblock
\shownote{We cite ePrint 2025/779, revision 2025-10-03; section, figure, and
  equation numbers refer to that version}.


\bibitem[Mizrahi~Erbes and Wattenhofer(2024)]%
        {MizrahiErbesWattenhofer2024Gather}
\bibfield{author}{\bibinfo{person}{Mose Mizrahi~Erbes} {and}
  \bibinfo{person}{Roger Wattenhofer}.} \bibinfo{year}{2024}\natexlab{}.
\newblock \showarticletitle{{Quit-Resistant Reliable Broadcast and Efficient
  Terminating Gather}}. In \bibinfo{booktitle}{\emph{28th International
  Conference on Principles of Distributed Systems (OPODIS 2024)}}
  \emph{(\bibinfo{series}{Leibniz International Proceedings in Informatics},
  Vol.~\bibinfo{volume}{324})}. \bibinfo{publisher}{Schloss Dagstuhl --
  Leibniz-Zentrum f{\"u}r Informatik}, \bibinfo{address}{Dagstuhl, Germany},
  \bibinfo{pages}{15:1--15:22}.
\newblock
\urldef\tempurl%
\url{https://doi.org/10.4230/LIPIcs.OPODIS.2024.15}
\showDOI{\tempurl}


\bibitem[Polyanskii(2026)]%
        {Polyanskii2026Vantage}
\bibfield{author}{\bibinfo{person}{Nikita Polyanskii}.}
  \bibinfo{year}{2026}\natexlab{}.
\newblock \bibinfo{title}{{Vantage: Availability-Graded Broadcast for
  Signature-Free BFT}}.
\newblock
\newblock
\showeprint[arxiv]{2608.16504v2}
\urldef\tempurl%
\url{https://arxiv.org/abs/2608.16504v2}
\showURL{%
\tempurl}
\newblock
\shownote{August 30, 2026}.


\bibitem[Raikwar et~al\mbox{.}(2024)]%
        {RaikwarEtAl2024SoKDAG}
\bibfield{author}{\bibinfo{person}{Mayank Raikwar}, \bibinfo{person}{Nikita
  Polyanskii}, {and} \bibinfo{person}{Sebastian M{\"u}ller}.}
  \bibinfo{year}{2024}\natexlab{}.
\newblock \showarticletitle{{SoK}: {DAG}-based Consensus Protocols}. In
  \bibinfo{booktitle}{\emph{IEEE International Conference on Blockchain and
  Cryptocurrency (ICBC)}}. \bibinfo{pages}{1--18}.
\newblock
\urldef\tempurl%
\url{https://doi.org/10.1109/ICBC59979.2024.10634358}
\showDOI{\tempurl}
\newblock
\shownote{Revised version: arXiv:2411.10026v2, April 22, 2025}.


\bibitem[Shrestha et~al\mbox{.}(2025a)]%
        {ShresthaEtAl2026OptimisticRBC}
\bibfield{author}{\bibinfo{person}{Nibesh Shrestha}, \bibinfo{person}{Qianyu
  Yu}, \bibinfo{person}{Aniket Kate}, \bibinfo{person}{Giuliano Losa},
  \bibinfo{person}{Kartik Nayak}, {and} \bibinfo{person}{Xuechao Wang}.}
  \bibinfo{year}{2025}\natexlab{a}.
\newblock \showarticletitle{{Optimistic, Signature-Free Reliable Broadcast and
  Its Applications}}. In \bibinfo{booktitle}{\emph{Proceedings of the 2025 ACM
  SIGSAC Conference on Computer and Communications Security (CCS '25)}}.
  \bibinfo{publisher}{ACM}, \bibinfo{address}{New York, NY, USA}.
\newblock
\urldef\tempurl%
\url{https://doi.org/10.1145/3719027.3765220}
\showDOI{\tempurl}
\showeprint[arxiv]{2505.02761v4}
\newblock
\shownote{Distinguished paper. We cite the revised arXiv version
  arXiv:2505.02761v4 (2026-05-24), which fixes a termination issue in the
  earlier design; theorem and definition numbers refer to v4}.


\bibitem[Shrestha et~al\mbox{.}(2025b)]%
        {ShresthaEtAl2025OptimisticRBCSuperseded}
\bibfield{author}{\bibinfo{person}{Nibesh Shrestha}, \bibinfo{person}{Qianyu
  Yu}, \bibinfo{person}{Aniket Kate}, \bibinfo{person}{Giuliano Losa},
  \bibinfo{person}{Kartik Nayak}, {and} \bibinfo{person}{Xuechao Wang}.}
  \bibinfo{year}{2025}\natexlab{b}.
\newblock \bibinfo{title}{{Optimistic, Signature-Free Reliable Broadcast and
  Its Applications}}.
\newblock
\newblock
\showeprint[arxiv]{2505.02761v2}
\urldef\tempurl%
\url{https://arxiv.org/abs/2505.02761v2}
\showURL{%
\tempurl}
\newblock
\shownote{arXiv:2505.02761v2; 2025-08-29; SUPERSEDED}.


\bibitem[Spiegelman et~al\mbox{.}(2022)]%
        {SpiegelmanEtAl2022Bullshark}
\bibfield{author}{\bibinfo{person}{Alexander Spiegelman}, \bibinfo{person}{Neil
  Giridharan}, \bibinfo{person}{Alberto Sonnino}, {and}
  \bibinfo{person}{Lefteris Kokoris-Kogias}.} \bibinfo{year}{2022}\natexlab{}.
\newblock \showarticletitle{Bullshark: {DAG BFT} Protocols Made Practical}. In
  \bibinfo{booktitle}{\emph{Proceedings of the ACM SIGSAC Conference on
  Computer and Communications Security (CCS)}}. \bibinfo{pages}{2705--2718}.
\newblock
\urldef\tempurl%
\url{https://doi.org/10.1145/3548606.3559361}
\showDOI{\tempurl}


\bibitem[Yu et~al\mbox{.}(2026)]%
        {YuEtAl2026SimpleIT}
\bibfield{author}{\bibinfo{person}{Qianyu Yu}, \bibinfo{person}{Juan Villacis},
  \bibinfo{person}{Giuliano Losa}, \bibinfo{person}{Zhuolun Xiang}, {and}
  \bibinfo{person}{Xuechao Wang}.} \bibinfo{year}{2026}\natexlab{}.
\newblock \bibinfo{title}{{Simple-IT: Practical Low-Latency Signature-Free BFT
  Consensus}}.
\newblock
\newblock
\showeprint[arxiv]{2606.14404v1}
\urldef\tempurl%
\url{https://arxiv.org/abs/2606.14404v1}
\showURL{%
\tempurl}
\newblock
\shownote{June 12, 2026}.


\end{thebibliography}

\appendix
\section{Complete six-execution lower-bound proof}
\label{app:totality-lower}

\subsection{Two steps are necessary}
\begin{lemma}[Folklore]
\label{lem:two-step-minimum}
In the model of Section~\ref{sec:model}, a reliable broadcast protocol tolerating
one Byzantine broadcaster cannot have an optimistic good-case latency
strictly below two steps, even if that latency is required only in
fault-free executions.
\end{lemma}
\begin{proof}
Take two fault-free executions with different broadcaster inputs $A,B$
and every interparty message delayed by exactly $\delta$. Before
$2\delta$, a non-broadcaster receives no input-dependent message from
another non-broadcaster: such a message cannot be sent before
$\delta$. All other non-broadcaster traffic received before the cut
is input-independent and identical in the two executions, including
timer-driven traffic. A Byzantine broadcaster can therefore reproduce
its $A$-execution stream to one nonempty group and its $B$-execution
stream to another. Keep every honest interparty delay equal to
$\delta$. Before $2\delta$, each recipient has precisely its local
history from the corresponding fault-free execution. If all honest
recipients had to deliver strictly before this cut, validity in the
two fault-free executions would force different deliveries in the
mixed execution, contradicting agreement.
\end{proof}

\subsection{Setup}
We give the joint schedule underlying Theorem~\ref{thm:totality-lower}
and prove Lemmas~\ref{lem:causal-cut} and~\ref{lem:common}. The
notation and partition are those of Section~\ref{sec:totality-lower}:
endpoint groups $E_A,E_B$ of size $s+1$, common groups $K_A,K_B$ of
size $f$, and a source group $\mathcal S$ of size $f-s-1$ containing broadcaster
$0$. Every member of $\mathcal S\setminus\{0\}$ is an ordinary recipient.
Choose the anchor $h\in \mathcal S\setminus\{0\}$ if $|\mathcal S|\ge2$, and $h=0$
otherwise. For the recipient-only interface, all occurrences of delivery
and totality below concern non-broadcasters. The choice $h\ne0$ is
available precisely when $n\le4f-2$; only $h=0$ requires the broadcaster to deliver.
Suppose a constant $D<3$ bounds totality delay whenever the broadcaster
is honest and at most $f$ non-broadcasters are faulty. Choose $\varepsilon,\eta>0$ with
\[
 2\varepsilon<\eta,\qquad H=2\varepsilon+D\eta<3\eta,\qquad J=4\eta.
\]
The code and initial states do not depend on the actual network delay
bound. With these fixed for the two inputs $A,B$, we may take $\eta$
small enough that $3\eta$ precedes every positive expiration of a timer
set during initialization. There are finitely many parties and finitely
many initialization events; if no such timer exists, the condition is
vacuous.

The construction never discards a packet or a timer: initial packets
and reactive timers are retained in every execution. Initial packets
are those sent during finite zero-time initialization before any
external receipt; non-broadcasters' initial packets are
input-independent and may be recipient-specific. Self-messages are
local events and are handled identically in compared histories.
Simultaneous events are ordered by one deterministic rule, depending
only on time, sender identity, and content, in all six executions. We
call $F_A$ and $F_B$ the \emph{seeds} of $X_A$ and $X_B$: the forcing
executions are built so that $h$ has its seed's history
through $2\varepsilon$.

\begin{table}[t]
\caption{Executions in the totality lower bound. The last column gives
the honest delivery forced by two-step optimism (fast executions), by the
hypothetical totality bound (forcing executions), or by copied histories
(bridges); ``$E_A{:}A$ by $H$'' means every member of $E_A$ delivers $A$
by time $H$.}
\label{tab:lower-worlds}
\small
\begin{tabular}{@{}llll@{}}
\toprule
Execution & Byzantine set & Largest honest delay & Forced delivery\\
\midrule
$F_A$ (fast) & $E_B$ & $\varepsilon$ & $h{:}A$ by $2\varepsilon$\\
$F_B$ (fast) & $E_A$ & $\varepsilon$ & $h{:}B$ by $2\varepsilon$\\
$X_A$ (forcing) & $K_B$ & $\eta$ & $E_A{:}A$ by $H$\\
$X_B$ (forcing) & $K_A$ & $\eta$ & $E_B{:}B$ by $H$\\
$W_A$ (bridge) & $\mathcal S\cup E_B$ & $J$ & $E_A{:}A$ by $H$\\
$W_B$ (bridge) & $\mathcal S\cup E_A$ & $J$ & $E_B{:}B$ by $H$\\
\bottomrule
\end{tabular}
\end{table}

\subsection{Fast seeds and forcing executions}
In $F_A$, the faulty group $E_B$ sends the initial packets its members
would send, with delay $\varepsilon$, and then remains silent. All honest
messages have delay $\varepsilon$. Define $F_B$ symmetrically, with $E_A$
faulty. The optimistic premise forces the honest anchor $h$ to deliver
$A$ in $F_A$ and $B$ in $F_B$ by $2\varepsilon$.

In $X_A$, delay all initial packets addressed to the honest group $E_B$
until $\eta$, including interparty links within $E_B$. Initial packets
to other honest parties arrive at $\varepsilon$, processed in the same
order as in $F_A$. In particular, $E_B$'s outgoing initial packets arrive
early at early recipients. All noninitial honest packets have delay
$\eta$, except packets addressed to $h$ and sent by time
$\varepsilon$, which have delay $\varepsilon$ (the anchor
exception). Faulty members of $K_B$ replay their $F_A$ streams to $h$
for sends through $\varepsilon$; honest $K_B$ in $W_A$ reacts to the
broadcaster's $B$ stream, so its early packets to $h$ would not match
$F_A$. Their other streams are fixed jointly with the bridges below.
$X_B$ is symmetric, with $E_A$ late and $K_A$ faulty.

Every honest party in $\mathcal S\cup E_A\cup K_A$ has its
seed history on
$[0,\varepsilon]$: only initial packets have arrived. Its receipts after
$\varepsilon$ differ from $F_A$, but any interparty packet it sends after
$\varepsilon$ reaches $h$ after $2\varepsilon$ in both executions.
When $h\ne0$, this reasoning includes broadcaster $0$, whose input and
initial receipts are identical in the seed and forcing execution. The
late group's initial outward packets are preserved, and its noninitial
packets cannot arrive at $h$ this early. The Byzantine seed copies supply
the remaining identities. Thus $h$ has exactly its seed history on
$[0,2\varepsilon]$, including zero-time reactions at the times $0$ and
$2\varepsilon$ themselves. It delivers $A$ in $X_A$ and $B$ in $X_B$ by
$2\varepsilon$. The first honest delivery is no later than this.
Validity and the assumed totality bound therefore force every member of
$E_A$ in $X_A$ to deliver $A$ by $H$, and every member of $E_B$ in $X_B$
to deliver $B$ by $H$.

\subsection{The joint bridge schedule}
The following rules complete $X_A,X_B,W_A,W_B$ simultaneously. When
copying a stream, a Byzantine party sends under its own identity and
preserves arrival times; no honest direct channel is forged. All initial
non-broadcaster packets arrive at $\varepsilon$ at honest bridge
recipients.

\begin{enumerate}
\item In both bridges, every faulty identity $z\in \mathcal S$ copies its own
$X_A$ stream to $E_A\cup K_A$ and its own $X_B$ stream to
$E_B\cup K_B$, preserving arrival times at honest recipients. This
includes $0$ and each ordinary recipient in $\mathcal S\setminus\{0\}$.
\item Noninitial packets sent by common parties in $K_A\cup K_B$ have
delay $\eta$ to honest recipients. In $W_A$, honest members of $E_A$
send with delay $\eta$ to $E_A\cup K_A$ and with delay $J$ to $K_B$. In
$W_B$, honest members of $E_B$ send with delay $\eta$ to $E_B\cup K_B$
and with delay $J$ to $K_A$. Packets to faulty recipients may use delay
$\eta$.
\item For each $q\in E_B$, its Byzantine copy in $W_A$ reproduces honest
$q$'s $W_B$ stream towards every common recipient, with exactly the same
arrival schedule. Towards each $p\in E_A$ it instead reproduces that
recipient's private stream from the late helper $q$ in $X_A$, again with
the same arrival schedule. Each $p\in E_A$ in $W_B$ does the symmetric
copying from $W_A$ towards common recipients and from $X_B$ towards
$E_B$.
\item In $X_A$, each faulty member of $K_B$ copies the stream of honest
$K_B$ in $W_A$ towards the early honest parties
$(\mathcal S\setminus\{h\})\cup E_A\cup K_A$, with initial packets arriving
at $\varepsilon$ and later packets taking $\eta$. To late recipients in $E_B$, its initial
packets arrive at $\eta$ and its subsequent copied packets have delay
$\eta$. To $h$, retain the seed-prefix exception for sends through
$\varepsilon$; later sends copy the $W_A$ counterpart with delay
$\eta$. These cases are disjoint: when $h\ne0$, source $0$ uses the
ordinary copied stream, not the exception. The faulty $K_A$ in $X_B$
behaves symmetrically.
\end{enumerate}

All honest-to-honest delays in the forcing executions are at most $\eta$,
and in the bridges at most $J$. A copy depends on an honest send already
determined at its send time. Every nonlocal delivery has a positive delay
at least $\varepsilon$. Fix initialization first and construct the worlds
by increasing time, or in successive $\varepsilon$-width intervals,
retaining the locally finite timer events and local reactions. Thus the
mutual copying has no dependence on a future receipt. The protocol may
embed entire histories or destination-specific information in a packet:
we copy complete packets and their event order, not summaries of them.

\subsection{Proof of the causal-cut lemma}
We prove simultaneously that $K_A$ has the same histories in $X_A,W_A$
before $2\eta$, $K_B$ in $X_B,W_B$ before $2\eta$, $E_A$ in
$X_A,W_A$ before $3\eta$, and $E_B$ in $X_B,W_B$ before $3\eta$. We
compare, sender identity by sender identity, the packets a party
receives before the cut in the two executions.

Consider $K_A$. In $X_A$, every initial interparty packet to the late group
$E_B$ is delayed, including those on links within $E_B$. No member of $E_B$
receives external information before $\eta$, and no initial positive
timer expires before $3\eta$. Hence it sends no noninitial packet
before $\eta$, and its first such packet reaches $K_A$ no earlier than
$2\eta$. In $W_A$, the copied opposite-bridge stream from $E_B$ to
$K_A$ has delay $J$ for noninitial packets. The initial packets match.
Packets from every member of $\mathcal S$ to $K_A$ are copied from $X_A$
into $W_A$ by rule~1,
and $K_B$'s packets from $W_A$ into $X_A$ by rule~4. A packet from $E_A$
arriving at $K_A$ before $2\eta$ depends on an $E_A$ history before
$\eta$. Within-$K_A$ packets depend on strictly earlier histories. The
claim for $K_B$ is symmetric.

At an endpoint $p\in E_A$, packets from $\mathcal S$ and the direct private
$E_B$ stream are copied from $X_A$ into $W_A$ by rules~1 and~3, and
packets from $K_B$ from $W_A$ into $X_A$ by rule~4. A packet from $K_A$
arriving before $3\eta$ was sent before $2\eta$, during the matching
interval just established. A packet from another member of $E_A$ has
delay $\eta$ and depends on that member's strictly earlier history.
Thus induction on causal events establishes the entire $E_A$-group
match; there is no instantaneous internal group communication. The
$E_B$ argument is symmetric. These dependencies are all to strictly
earlier events, so the simultaneous induction is well founded rather
than circular.

The statements concern complete timed histories. Determinism therefore
preserves timer schedules, every local state transition, and messages
triggered by delivery. A short reactive timer cannot precede its
triggering receipt. The construction does not discard any packet or
stop a party at delivery. The cut is strict: it makes no claim that
endpoint histories still agree at $3\eta$.

Since $H<3\eta$, all forced $E_A$ deliveries are inherited in $W_A$ and
all forced $E_B$ deliveries in $W_B$. In particular, the bridges contain
opposite honest deliveries.

\subsection{Indefinite common histories}
Every member of $K_A\cup K_B$ receives the same stream from each
identity in $\mathcal S$ in both bridges. For each identity in $E_A$, it receives an honest $W_A$
stream and its Byzantine $W_B$ copy; for each identity in $E_B$, an
honest $W_B$ stream and its Byzantine $W_A$ copy. Common-to-common
packets have identical schedules and, inductively, identical contents.
Therefore the common parties have identical complete histories forever,
identity by identity, which is Lemma~\ref{lem:common}. The differing
private $E_A$--$E_B$ streams impose no conflicting channel obligation:
every such channel has a Byzantine party at one end in each bridge.

Continue the causal construction indefinitely, preserving every queued
packet. Honest bridge delays remain at most $J$, so the executions are
fair. Agreement and eventual totality force each common honest party to
deliver $A$ in $W_A$ and $B$ in $W_B$. Identical deterministic histories
make this impossible. We have contradicted every constant $D<3$ already
on honest-broadcaster executions, proving Theorem~\ref{thm:totality-lower}.

\section{Paired-witness details and fast-echo baseline}
\label{app:paired-details}

\subsection{Bracha with fast echo delivery}
\label{sec:fast-bracha}
The following elementary baseline separates honest-broadcaster spread
from Byzantine-broadcaster recovery; it is a modification of
Bracha~\cite{Bracha1987}, not an additional paired-witness result.
Let $N=n-1$, $Q=N-b$, and $E=\lceil(N+f)/2\rceil$.
Non-broadcasters echo the first proposal. Every party sends ready on
$E$ echoes or $f+1$ readies, and delivers on $Q$ echoes or $2f+1$
readies. Sends are one-shot across values, slots are immutable, and
parties continue processing after delivery.

\begin{proposition}[Fast-echo baseline]
\label{prop:fast-bracha}
For $3f+1\le n<4f$, $s=n-3f-1$, and
$1\le b\le\lfloor s/2\rfloor+1$, this protocol has quadratic messages
and guarantees $(2,3,3)$, with $d_t^{\mathrm B}\le2$.
At $n=3f+1,b=1$, Theorem~\ref{thm:totality-lower} gives
$d_t^{\mathrm H}=d_t=3$.
\end{proposition}
\begin{proof}
With a Byzantine broadcaster there are at most $t=f-1$ faulty echo
authors. The stated range gives
$2Q-N>t$, $Q+E-N>t$, and $2E-N>t$.
Thus two fast quorums, a fast quorum and a ready-origin echo quorum,
or two ready-origin echo quorums share an honest echo author. These
intersections and the $f+1$ ready relay ensure agreement by the usual
first-honest-ready argument. With an honest broadcaster, $E,Q>f$
ensure validity. The $N-b$ honest echoes suffice for two-step delivery
under $b$ faults; $E\le N-f$ gives ready and delivery by step three
under all $f$ faults.

After a fast delivery at $T$ with a Byzantine broadcaster,
$Q-t\ge N-b-f+1\ge E$ honest echoes are already sent; the last
inequality is exactly $b\le\lfloor s/2\rfloor+1$ at $t=f-1$.
They reach everyone by $T+\delta$, causing ready, and delivery follows
by $T+2\delta$. A ready delivery leaves $f+1$ honest readies in
transit and has the same bound by relay. Safety excludes earlier
opposite readies. These arguments also give eventual totality.
For an honest broadcaster, $d_t^{\mathrm H}\le d_s\le3$.
Each party sends at most an echo and a ready multicast, besides the
source proposal, proving the message bound, including pre-corruption
and post-delivery sends.
\end{proof}

The baseline reaches the lower-bound premise $b\ge s+1$ only at $s=0$.
Thus the seven-party matching instance already follows from this simple
protocol. Paired witnesses support larger budgets, including two and
three at minimum resilience when these are admissible.

\subsection{Paired-witness thresholds and message flow}

\begin{table}[t]
\caption{Thresholds of Algorithm~\ref{alg:paired}, with $N=n-1$ and
$t\le f-1$ Byzantine non-broadcasters when the broadcaster is Byzantine.}
\label{tab:paired-thresholds}
\small\setlength{\tabcolsep}{4pt}
\begin{tabular}{@{}lp{\dimexpr\columnwidth-4.55cm\relax}@{}}
\toprule
Threshold (use) & Why this value\\
\midrule
$Q=N-b$ (echo delivery) & every honest echo has arrived after two
  steps when at most $b$ parties are faulty\\
$C=N-f$ (confirm, its delivery) & every honest echo, or confirm,
  has arrived when at most $f$ parties are faulty\\
$M=Q-f+1$ (vote) & honest echoes guaranteed behind an echo delivery:
  $Q-t\ge M$\\
$P=C-f+1$ (witness) & honest confirmers guaranteed behind a confirm
  delivery: $C-t\ge P$\\
$R=N-f+1$ (ready) & votes of all honest non-broadcasters:
  $N-t\ge R$\\
$R_{\mathrm w}=2f-1$ (ready) & at least $R_{\mathrm w}-t\ge f$ honest witnesses; $R_{\mathrm w}<R\le N-t$,
  so honest witnesses alone reach $R_{\mathrm w}$ after a confirm delivery\\
\bottomrule
\end{tabular}
\end{table}

\begin{figure}[t]
\centering
\begin{tikzpicture}[>=Latex,font=\scriptsize,
  m/.style={draw,rounded corners=2pt,minimum height=4.5mm,minimum width=9.5mm,inner sep=2pt},
  d/.style={align=center,text=black!70}]
\node[m] (p) at (0,0) {\msg{propose}};
\node[m] (e) at (1.7,0) {\msg{echo}};
\node[m] (c) at (3.5,0.75) {\msg{confirm}};
\node[m] (v) at (3.5,-0.75) {\msg{vote}};
\node[m] (w) at (5.3,0.75) {\msg{witness}};
\node[m] (r) at (7.0,0) {\msg{ready}};
\draw[->] (p) -- (e);
\draw[->] (e) -- node[above,sloped,pos=.55]{$C$} (c);
\draw[->] (e) -- node[below,sloped,pos=.55]{$M$} (v);
\draw[->] (c) -- node[below,yshift=-1mm,align=center]{$P$ authors with\\echo and confirm} (w);
\draw[->] (w) -- node[below,sloped,pos=.45]{$R_{\mathrm w}$} (r);
\draw[->] (v) -- node[below,sloped,pos=.45]{$R$} (r);
\draw[->] (r.south) to[out=-60,in=-120,looseness=5] node[below]{$f{+}1$} (r.south);
\node[d,below=1.2mm of e] {deliver on\\$Q$ echoes};
\node[d,above=1.2mm of c] {deliver on\\$C$ confirms};
\node[d,above=1.2mm of r] {deliver on\\$2f{+}1$ readies};
\end{tikzpicture}
\caption{Message flow of Algorithm~\ref{alg:paired}. Arrow labels give
the number of distinct authors that trigger the next message; the
witness rule counts authors whose echo and confirm slots both hold $v$,
so it reads echo slots as well as confirm slots. Confirm also requires
the author's own echo. With an honest
broadcaster every honest party has $Q$ echoes after two steps when at
most $b$ parties are Byzantine, and $C$ confirms after three. After a
first honest delivery, the vote path (behind an echo delivery) or the
witness path (behind a confirm delivery) produces readies within two
steps and delivery within three.}
\Description{Propose leads to echo; echo leads to confirm and to vote;
confirm paired with echo leads to witness; witness and vote lead to
ready; ready amplifies itself. Delivery happens on echo, confirm, or
ready quorums.}
\label{fig:paired-flow}
\end{figure}
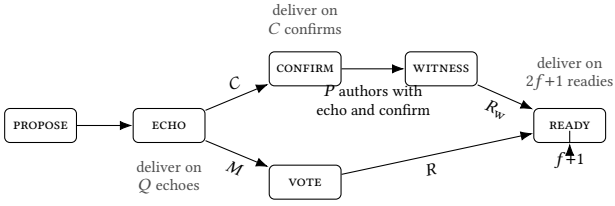

\paragraph{Design rationale}
Suppose the broadcaster is Byzantine and an honest party delivers $v$.
If it delivered on $Q$ echoes, at least $Q-t\ge M$ honest parties echoed
$v$; within one step every honest party holds those echoes and votes,
the votes trigger ready, and ready triggers delivery, three steps in
all. If it delivered on $C$ confirms, at least $C-t\ge P$ honest parties
confirmed $v$, and the natural analogue would let $P$ confirms trigger
ready. That rule is unsafe. Take $t=f-1$ and $b\ge s+2$. The $M-t$
honest echoes behind a vote for $v$ and the $C-t$ honest echoes behind a
confirm for $w\ne v$ fit disjointly among the $N-t$ honest
non-broadcasters, because $(M-t)+(C-t)\le N-t$ is exactly $b\ge s+2$;
Byzantine echo slots complete both quorums; and the $P-t$ honest
confirmers of $w$ leave $N-P=2f-1\ge R-t$ other honest parties free to
vote for $v$, since $b\ge s+2$ forces $f\ge s+3$. So $R$ votes for $v$
and $P$ confirms for $w$ can both be assembled before any delivery.
Opt-RBC inserts an acknowledgment round between votes and
ready~\cite{ShresthaEtAl2026OptimisticRBC}, and our confirm/ack protocol
(Appendix~\ref{sec:baseline-results}) routes both paths through such a
round; either way the acknowledgment costs a step. The witness message
instead certifies, at its own author, $P$ identities that have both
echoed and confirmed $v$. Lemma~\ref{lem:paired-compatibility} shows
that a quorum of such witnesses rules out any $R$-vote quorum for
another value, so witnesses may trigger ready directly.

The budget cap $b\le\lfloor n/2\rfloor-f$, which Theorem~8 of Shrestha
et al.\ shows necessary for two-step optimism, enters precisely in
$Q+M-N>t$. The hypothesis $b\le2s+3$ enters precisely in the displayed
inequality of Lemma~\ref{lem:paired-compatibility}; it is a restriction
of this construction, not an impossibility bound on optimistic
broadcast.

The lower endpoint $f=2$ is permitted throughout the construction.
The validity argument requires $R_{\mathrm w}=2f-1>f$, which fails at $f=1$:
one faulty non-broadcaster can send a witness for a value different from
an honest broadcaster's input and trigger honest readies for that value.
Thus the present rules do not extend to $f=1$.

\subsection{Failure beyond the paired-witness budget}
\begin{proposition}[Boundary of the fixed rules]
\label{prop:paired-boundary}
Keep Algorithm~\ref{alg:paired} and thresholds~\eqref{eq:paired-thresholds}.
For every $2s+4\le b\le\min\{f,\lfloor n/2\rfloor-f\}$, these rules
admit an execution violating totality. Thus increasing their budget
beyond $2s+3$ requires changing the protocol, not only its proof.
\end{proposition}
\begin{proof}
The range implies $f\ge3s+7$. Corrupt the source and $f-1$
non-broadcasters. Partition the $2f+s+1$ honest non-broadcasters into
$U$ of size $f+s+1$ and $V$ of size $f$, and choose $U_0\subset U$ of
size $f-1$. The source proposes $v$ to $U$ and $w\ne v$ to $V$.
Partition faulty non-broadcasters into $Y$ of size $s+2$ and $X$ of
size $f-s-3$. Authors in $Y$ echo and confirm $v$ to everyone; authors
in $X$ echo $v$ to $U_0$ and $w$ elsewhere.

Delivering the echoes makes $U_0$ hold $C$ echoes for $v$, vote and
confirm $v$. Outside $U_0$, there are $f+|X|=2f-s-3\ge M$ echoes for
$w$ and only $f+2s+3<C$ for $v$, so schedule their votes for $w$.
There are $f+s+2$ such honest voters. Supply the $f-1$ Byzantine
$w$-votes to them; they reach $R$ and send ready for $w$.
Delay those readies towards $U_0$. Next deliver the confirms: every
honest party has $P$ paired authors in $U_0\cup Y$ and witnesses $v$.
These witnesses make $U_0$ send ready for $v$ before the delayed readies
arrive. All honest parties have now spent their unique ready: $f-1$
for $v$ and $f+s+2$ for $w$.

Give Byzantine readies for $w$ to just one honest party outside $U_0$;
it reaches $2f+1$ readies and delivers. Deliver all remaining honest
traffic and let faulty parties remain silent. Other parties cannot
reach $2f+1$ readies for either value, since $f+s+2<2f+1$.
No other delivery rule applies: $Q>C$, every echo count is at most
$C$, and every confirm count is at most $P<C$. Exactly one honest
party delivers. All delays used before draining can be finite, so this
is an eventual-totality failure. The first instance is $(22,7,4)$.
\end{proof}

\section{Coded paired-witness broadcast}
\label{app:paired-coded}

We prove Proposition~\ref{prop:paired-coded}. Use a deterministic $[N,k]$
MDS encoder over $\mathrm{GF}(2^{\omega})$, where $k=P=n-2f$ and $2^{\omega}>N$.
Bind the instance identifier, payload length, code parameters, and padding
rule into a domain-separated Merkle commitment to the ordered fragment
vector. A fragment has $F=w\lceil L/(kw)\rceil$ bits. A successful
reconstruction erasure-decodes $k$ valid distinct-index fragments, checks
padding, re-encodes the complete vector, and rebuilds the received root.
Individually valid Merkle paths do not establish that a vector is a codeword.
All claims below are conditional on commitment binding.
Algorithm~\ref{alg:paired-coded} gives the event loop; the packet rules
below define admission and reference resolution. Counters $e,v,c,w,r$
count distinct authors of the indicated message type and root, and $p$
counts matched echo/confirm authors. Only echoes with available fragments
enter $e$ or $p$.

\begin{algorithm}[t]
\DontPrintSemicolon
\caption{Coded paired witnesses at party $i$. Thresholds are those of Algorithm~\ref{alg:paired}.}
\label{alg:paired-coded}
\textbf{State:} immutable per-author/type root slots; first valid source
preload; fragment cache; tested roots and successful decodings; per-type
sent flags; own-fragment multicast flags indexed by root.\;
\textbf{Source initialization:} encode the value, commit to its fragment
vector, and send each non-source $j$ its indexed fragment and proof.\;
\textbf{On a source preload:} accept only the first valid packet for
index $i$; cache it and, if $i$ is a non-source, emit echo for its root.\;
\textbf{On an eligible control packet:} reserve its author's first
well-formed root for its type; ignore conflicting later roots. For
echo/vote only, admit an attached fragment for this reserved root, at the author's index,
with a valid proof. Repeated matching packets can supply missing data.\;
\textbf{After either event:} resolve references from the fragment cache
and update counters. For each affected root $\rho$, apply the following
enabled rules until none remain.\;
Disable each send rule once its type was sent, and disable the delivery
rule after delivery.\;
\Indp
If $\rho$ has $k$ distinct cached fragments and is untested, erasure-decode,
check padding, re-encode and compare roots; cache success or failure.\;
If $i$ is a non-source and $e(\rho)\ge M$ and $\rho$ decoded successfully,
emit vote for $\rho$.\;
If $i$ is a non-source and $e(\rho)\ge C$, its own echo is $\rho$, and $\rho$
decoded successfully, emit confirm for $\rho$.\;
If $i$ is a non-source and $p(\rho)\ge P$, emit witness for $\rho$.\;
If $v(\rho)\ge R$ or $w(\rho)\ge R_{\mathrm w}$ or $r(\rho)\ge f+1$, emit ready for $\rho$.\;
If $\rho$ decoded successfully and at least one of $e(\rho)\ge Q$,
$c(\rho)\ge C$, or $r(\rho)\ge2f+1$ holds, deliver its payload once.\;
\Indm
\textbf{Emit $(\tau,\rho)$:} if type $\tau$ was already sent, do nothing.
Otherwise mark it sent and multicast. Echo/vote attaches $i$'s fragment
and proof unless already multicast for $\rho$, in which case it references
that fragment. Confirm/witness/ready carries only the root.\;
Count eligible self-messages immediately. Continue after delivery.\;
\end{algorithm}

\paragraph{Packet rules}
Use Algorithm~\ref{alg:paired} with roots as control values. The broadcaster
preloads each non-broadcaster with that recipient's own-index fragment
and proof. Keep only the first valid preload from the authenticated
broadcaster. Echo only after validating this packet. An echo supplies its author's
own-index fragment and proof, or references the identical fragment already
multicast with its vote. Reserve the first well-formed echo notice per
author; count it only once its matching valid fragment is available.
A conflicting later root never replaces the reserved notice.

Vote requires $M$ counted echoes and successful reconstruction. It supplies
the voter's re-encoded own-index fragment and proof unless that exact
fragment was already multicast in its echo. Confirm requires $C$ counted
echoes, the own-echo guard, and successful reconstruction. Witness and ready
retain their original rules and carry roots only; they need not wait for
local reconstruction. Every delivery additionally requires successful
reconstruction and returns the resulting payload.

Deduplication works in both directions: a party can vote before its broadcaster
preload arrives and echo later. Keep a multicast-sent flag for each
own-index/root pair, separately from the local cache. Buffer references
that overtake their data and reconsider them when data arrives. An honest
echo and its fragment were both sent by the echo's send time, so both
arrive within one delay of that time. Thus counted echo receipt preserves
the full-value delay bound. A vote notice may be counted before its
referenced fragment arrives. Reserve the first well-formed vote notice
per author, as for echo. Admit a received fragment only for the retained
broadcaster preload or a reserved echo/vote root. A preload uses the
recipient's index; an echo/vote fragment uses its authenticated author's
index. Require fixed length, valid symbols, and a valid proof, and cache
at most one fragment per index/root pair. A fragment-bearing echo/vote
packet also supplies its notice, so data that overtakes a reference can
reserve the corresponding slot. Use immutable first-root slots for the
other message types as well.

\paragraph{Safety and malformed codewords}
The full-value safety proof applies to the root-level first-message ledgers;
extra reconstruction guards restrict honest sends without enabling any new
one. A noncodeword commitment cannot pass reconstruction and full root
checking. It cannot obtain an honest vote or confirm. Every honest witness
has an honest confirmer among its paired authors. Every honest ready
origin has an honest vote or an honest witness behind it; ready relays
preserve this fact by causal induction. Thus a noncodeword cannot originate
honest ready, and it cannot be delivered on raw echoes because delivery
checks reconstruction separately. For a valid root, binding and encoder
injectivity make all successful decoders return the same payload. Root
agreement therefore gives payload agreement.

\paragraph{Availability at every deadline}
Since $M-k=f-b\ge0$ and $C,Q\ge k$, honest-broadcaster reconstruction
requires no extra message step. The two- and three-step good cases follow
as in Lemma~\ref{lem:paired-latency}.

For a Byzantine broadcaster, let $t\le f-1$ be the faulty non-broadcaster
count and $T$ the first honest delivery time. An echo delivery leaves
$Q-t\ge M\ge k$ honest own-index fragments already multicast. By
$T+\delta$ everyone has them, can reconstruct and vote, and the usual
vote/ready chain completes by $T+3\delta$.
A confirm delivery leaves $C-t\ge P=k$ honest confirmers, each of whom
already echoed its own fragment. Their fragments and paired messages
reach everyone by $T+\delta$, enabling reconstruction and witness;
delivery follows by $T+3\delta$.

For a ready delivery, trace the first honest ready of its root to a
nonrelay origin at time $T_{\mathrm{orig}}\le T$. A vote origin has at least
$R-t\ge k+1$ honest vote authors, each of whom already multicast its
own-index fragment. A witness origin has an honest witness and hence an
honest confirmer; that confirmer's $C$ echoes contain at least $C-t\ge k$
honest own-index fragments already multicast. In either case $k$ honest
fragments were sent before $T_{\mathrm{orig}}$. They reach every honest party by
$T_{\mathrm{orig}}+\delta\le T+\delta$, so the ready-relay delivery by $T+2\delta$
is not delayed by payload reconstruction. The argument also proves
asynchronous eventual availability without a known delay bound.

\paragraph{Accounting}
There are $N$ broadcaster fragments. With an honest broadcaster, every honest
echo and vote has the same root, so symmetric deduplication permits one
fragment multicast per non-broadcaster. Payload traffic is at most
$NF+N^2F=nNF$. Universally, each non-broadcaster can multicast at most
two fragments, one for each of its one-shot echo/vote roots. Even charging
an honest preload before a later broadcaster corruption gives
$N(2N+1)F$. Since $N/k<3$, the leading coefficients are below $3nL$
and $6nL$. Merkle proofs, control roots, and $F-L/k<\omega=O(\log n)$ padding
add $O(\kappa n^2\log n)$ bits. No repair message type is added; the
message cap remains $5N^2+2N$ and includes all post-delivery traffic.

First-root discipline permits at most $2N+1$ accepted distinct fragment
entries at a receiver. Hence fewer than six roots can collect $k$ entries,
since $(2N+1)/k<6$. Cache failed as well as successful reconstruction
checks: a failure for a binding root cannot be repaired by receiving more
fragments of that same root. At most five full decoding/re-encoding
attempts per receiver suffice, in addition to per-packet proof checks.
This bounds work on admitted records and the number of expensive
reconstructions; proof validation is polynomial per packet. It does not
bound total processing of an unbounded Byzantine packet stream or
measure decoding time in practice.

\paragraph{Communication baselines}
The linear payload order is established prior work. Opt-RBC v4 already
gives a balanced coded protocol with
$O(nL+\kappa n^2\log n)$ bits~\cite[Section~5 and Figure~3]{ShresthaEtAl2026OptimisticRBC}.
At $n=3f+1$, its outer dimension is $f+1$; counting both echo and vote
fragment multicasts as written gives a leading upper bound $6nL$ plus
$O(L)$ payload. Deduplicating identical echo/vote fragments reduces the
honest-broadcaster leading bound to $3nL$, while retaining the universal
$6nL$ upper bound. These are our direct counts of those rules, not
communication lower bounds or additional source theorems. Our corresponding
$3nL/6nL$ bounds therefore do not improve that optimized payload comparison;
the gain is preserving the paired-witness latency guarantees.

Other hash-only broadcasts have smaller metadata terms. Das, Xiang, and
Ren report $7nL+2\kappa n^2+2n^2$ bits for their RBC, including its
four-round version~\cite[Section~8]{DasXiangRen2021ADD}.
CC-BRB reports $3nL+9\kappa n^2$, omitting lower-order terms, with a
three-step honest-broadcaster path~\cite[Table~I and Section~V]{AlhaddadEtAl2022HashBRB}.
These source-reported constants use minimum resilience and their own
accounting conventions; they are not claims of two-step optimistic
delivery. Our $O(\kappa n^2\log n)$ Merkle remainder is larger
asymptotically. The purpose of this appendix is timely availability for
$(2,3,3)$, not an improvement to the communication frontier.

\section{Consensus latency consequences}
\label{app:consensus-applications}
Reliable broadcast latency enters a consensus protocol in two ways: the
good-case latency bounds how long an honest proposal takes to arrive,
and the totality delay bounds how far honest parties can drift apart,
which sizes timeouts. Two existing signature-free protocols state
their latencies as functions of the broadcast they use, so the effect
of substituting Algorithm~\ref{alg:paired} for Opt-RBC can be read off.
Table~\ref{tab:applications} does so. The Sailfish++ formulas are from
Section~7 of Shrestha et al.~\cite{ShresthaEtAl2026OptimisticRBC} and
the two non-leader classes of Yu et al.'s Table~2; the Simple-IT
formulas are from its Lemmas~6--8 and~12 and
Corollaries~3--5~\cite{YuEtAl2026SimpleIT}. Following Simple-IT,
$\delta$ is the actual network delay bound after GST and $\Delta>\delta$ its
known bound; the broadcast protocols use neither. The entries in the
last column are our substitutions, not results reported by those
authors, and every commit path assumes the required honest leaders,
resolved dependencies, and no intervening timeout.

\begin{table*}[t]
\caption{Consensus latency with three reliable broadcast protocols.
Commit and block-time rows are multiples of $\delta$; skew, timer, and view
rows are as written. In an $a/b$ entry, $a$ is the optimistic value with
$d=d_o$ and $b$ the good-case value with $d=d_s$. Starred optimistic
entries assume the predecessor-readiness premise stated in the text. The
paired-witness column requires $f\ge2$ and has budget $\min\{b_*,2s+3\}$
in place of Opt-RBC's $b_*$; at $n=3f+1$ this is at most
$\min\{3,\lfloor(f+1)/2\rfloor\}$ faults.}
\label{tab:applications}
\small\setlength{\tabcolsep}{4pt}
\centering
\begin{tabular}{@{}lrrr@{}}
\toprule
Metric (formula) & Bracha & Opt-RBC & Paired witness\\
\midrule
RBC $(d_o,d_s,d_t)$ & $(3,3,2)$ & $(2,4,4)$ & $(2,3,3)$\\
\midrule
\multicolumn{4}{@{}l}{Sailfish++ commit paths, $n=3f+1$}\\
Leader, rule 1 ($d+1$) & 4 & 3 / 5 & 3 / 4\\
Leader, rule 2 ($2d$) & 6 & 4 / 8 & 4 / 6\\
Non-leader, next-leader path ($2d+1$) & 7 & 5 / 9 & 5 / 7\\
Non-leader, later-reference path ($3d+1$) & 10 & 7 / 13 & 7 / 10\\
\midrule
\multicolumn{4}{@{}l}{Base Simple-IT (propose upon delivery)}\\
Commit ($d_o+1$ / $d_s+1$) & 4 & $3^*/5$ & $3^*/4$\\
Optimistic block time ($d_o$) & 3 & $2^*$ & $2^*$\\
Round-entry skew ($d_t\delta$) & $2\delta$ & $4\delta$ & $3\delta$\\
Round timer ($(d_s+d_t)\Delta$) & $5\Delta$ & $8\Delta$ & $6\Delta$\\
View duration ($\Delta_{\mathrm{to}}+d_t\delta$) & $5\Delta+2\delta$ & $8\Delta+4\delta$ & $6\Delta+3\delta$\\
\bottomrule
\end{tabular}
\end{table*}

\paragraph{Sailfish++}
Sailfish++ commits the leader vertex of round $r$ under its first rule
after one RBC plus one step, the step being the arrival of the first
RBC message of $2f+1$ round-$(r+1)$ vertices, and under its second rule
after two RBCs~\cite[Section~7]{ShresthaEtAl2026OptimisticRBC}.
Non-leader vertices referenced by the next committed leader need one
further RBC, and the remaining ones two~\cite[Table~2]{YuEtAl2026SimpleIT}.
In this honest-broadcast commit path ``one RBC'' is the optimistic or
good-case latency, not the totality delay, and in
Algorithm~\ref{alg:paired} the first RBC message is the proposal itself,
as in Opt-RBC. This first-message argument uses the full-proposal
interface; a coded instantiation must additionally expose the DAG
metadata required by the commit rule. Substituting $(2,3,3)$ lowers good-case leader commit
from five steps to four and the two non-leader paths from nine and
thirteen steps to seven and ten, while the optimistic entries are
unchanged. This is path accounting for honest broadcasts; it is not a
bound for arbitrary Byzantine vertices.

The source's timeout remark and Claim~23 analyze Bracha explicitly:
$5\Delta$ suffices after entry via the preceding leader's delivery,
and $8\Delta$ after timeout-based entry~\cite[Section~7 and
Appendix~F, Claims~19--23]{ShresthaEtAl2026OptimisticRBC}.
These constants include both round-entry skew and subsequent honest
broadcasts. They are not generic timeout formulas for an arbitrary RBC.
We therefore do not substitute our latency triple into those numerical
constants or claim a Sailfish++ timeout improvement here; the timer and
view-duration improvements in Table~\ref{tab:applications} concern
Simple-IT, whose analysis explicitly parameterizes the RBC.

\paragraph{Simple-IT}
Lemma~12 of Yu et al.~\cite{YuEtAl2026SimpleIT} bounds round-entry skew:
if an honest party enters round $r$ at time $t\ge\mathrm{GST}$,
every honest party has entered that round by $t+d_t\delta$.
Entry requires each preceding round to be marked safe through proposal
delivery or disabled through timeout notification. The evidence already
obtained by the first party reaches all others within the larger of
the two totality delays. Timeout notification takes at most $2\delta$,
so $d_t\ge2$ makes RBC totality the controlling bound. With paired
witnesses, this gives $3\delta$ in place of Opt-RBC's $4\delta$.
This bounds the time to reach the same round, not the difference in
round numbers or an unconditional rate of DAG growth.

Simple-IT's base variant proposes upon delivery and sets its round
timer to $\Delta_{\mathrm{to}}=(d_s+d_t)\Delta$, with eventual worst-case
view duration $\Delta_{\mathrm{to}}+d_t\delta$~\cite[Lemma~8 and
Corollary~5]{YuEtAl2026SimpleIT}. Substituting $(2,3,3)$ gives $6\Delta$
and $6\Delta+3\delta$ in place of $8\Delta$ and $8\Delta+4\delta$;
Bracha gives $5\Delta$ and $5\Delta+2\delta$. For the optimistic entries
we state the premise under which we claim them: every honest party can
validate the proposal's predecessor chain within $2\delta$ of the
honest proposal time, for instance because the chain is already
resolved, and the round is eligible and has not timed out; block time
additionally assumes consecutive honest leaders. Under this premise
commit takes $(d_o+1)\delta=3\delta$ and block time
$d_o\delta=2\delta$~\cite[Lemma~6 and Corollary~4]{YuEtAl2026SimpleIT}.
Without it, combining proposal delivery with predecessor propagation
gives $(\max\{d_o,d_t\}+1)\delta=4\delta$ for paired witnesses, which is
also the good-case bound since $d_s=d_t=3$. The same premise is marked
for the Opt-RBC column. The two-step block time concerns the base
variant, not the speculative pipelining of its Lemma~9.

\paragraph{Optimistic budgets}
Simple-IT's optimistic fault budget $f_o=\lfloor n/2\rfloor-f$ (its
Definition~9 and Table~3) equals $b_*$ and corresponds to
$\lceil n/2\rceil+f$ correct parties including the leader. Paired
witnesses keep this budget where $b_*\le2s+3$ and reduce it otherwise.

\section{Confirm/ack construction and analysis}
\label{sec:baseline-results}

This appendix gives our confirm/ack construction, which retains maximal
optimism for every $n\ge3f+1$. It uses a separate ack layer and can require
four steps after the first honest delivery. We give the construction,
prove safety and latency, and then establish its exact fault-dependent
totality profile at $n=3f+1$, including smaller optimistic budgets.

\subsection{Construction and thresholds}
\label{sec:protocol}

Let $n\ge3f+1$, $f\ge1$, and $N=n-1$.  We describe one protocol with three
thresholds: $M$ for votes, $C$ for confirms and for the ack and ready
quorums, and $Q$ for delivery on echoes.  Fix an integer optimistic
threshold $0\le b\le f$ and assume
\begin{equation}
\label{eq:general-inequalities}
\begin{aligned}
 f<M\le C\le N-f,&\qquad 2C\ge N+f,\\
 Q+M\ge N+f,&\qquad Q-(f-1)\ge M,\\
 C+M>N,&\qquad C-(f-1)\ge f+1,
\end{aligned}
\qquad Q\le N-b.
\end{equation}
These are the hypotheses of the safety and latency lemmas in
Sections~\ref{sec:safety}--\ref{app:totality-tightness}; the last condition provides enough honest echoes for
two-step delivery with at most $b$ Byzantine parties.  They imply $Q\ge M>f$
and $2Q\ge Q+M\ge N+f$.  Table~\ref{tab:thresholds} says what each
inequality is used for.

\begin{lemma}[Admissible thresholds]
\label{lem:admissible}
For every $f\ge1$, $n\ge3f+1$, and integer
$0\le b\le\min\{f,\lfloor n/2\rfloor-f\}$, the thresholds
\begin{equation}
\label{eq:general-thresholds}
 C=\left\lceil\frac{N+f}{2}\right\rceil,\qquad Q=N-b,\qquad
 M=\min\{C,\,n-f-b\}
\end{equation}
satisfy~\eqref{eq:general-inequalities}.  At $n=3f+1$ and
$b=\lfloor(f+1)/2\rfloor$ they equal $C=2f$, $Q=3f-b$, and $M=2f+1-b$.
\end{lemma}
\begin{proof}
By definition $M\le C$, $2C\ge N+f$, and $Q=N-b$.  If $N=3f$ then
$C=2f=N-f$; otherwise $C\le(N+f+1)/2\le N-f$ because $N\ge3f+1$.  Also
$C\ge2f$, hence $C-(f-1)\ge f+1$.  Since $C\ge2f>f$ and
$n-f-b\ge n-2f\ge f+1$, we have $M>f$, including when $b=0$.
The bound on $b$ gives $C\ge f+b$ and $n-f-b\ge f+b$, so
$M\ge f+b$ and $Q+M\ge(N-b)+(f+b)=N+f$.  Also
$M\le n-f-b=Q-f+1$.  If $M=C$ then $C+M=2C\ge N+f>N$; otherwise
$M=n-f-b$ and $C+M-N=C-f-b+1\ge1$.  The values at $n=3f+1$ follow by
substitution.
\end{proof}

Algorithm~\ref{alg:rbc} lists the rules and Figure~\ref{fig:flow} shows the
message flow.  Only non-broadcasters send \msg{echo}, \msg{vote},
\msg{confirm}, and \msg{ack}; these quorums count distinct non-broadcasters.
Every party, including the broadcaster, sends \msg{ready} and delivers;
ready quorums count all $n$ parties.  Each rule fires at most once per
party, and ``send'' means send to all parties.  The examples below use
$n=3f+1$ with the maximal $b$.

\begin{algorithm}[t]
\DontPrintSemicolon
\caption{Reliable broadcast with thresholds satisfying~\eqref{eq:general-inequalities}.}
\label{alg:rbc}
\textbf{Broadcaster} with input $v$: send $\m{propose}{v}$.\;
\medskip\textbf{Non-broadcaster:}\;
\Indp
\nl\textbf{upon} the first $\m{propose}{v}$ from the broadcaster:
  send $\m{echo}{v}$.\;
\nl\textbf{upon} $\m{echo}{v}$ from $M$ non-broadcasters:
  send $\m{vote}{v}$.\;
\nl\textbf{upon} $\m{echo}{v}$ from $C$ non-broadcasters:
  send $\m{confirm}{v}$.\;
\nl\textbf{upon} $\m{vote}{v}$ from $C$ non-broadcasters,
  or $\m{confirm}{v}$ from $f{+}1$ non-broadcasters:
  send $\m{ack}{v}$.\;
\Indm
\medskip\textbf{Every party:}\;
\Indp
\nl\textbf{upon} $\m{ack}{v}$ from $C$ non-broadcasters,
  or $\m{ready}{v}$ from $f{+}1$ parties:
  send $\m{ready}{v}$.\;
\nl\textbf{upon} $\m{echo}{v}$ from $Q$ non-broadcasters,
  or $\m{confirm}{v}$ from $C$ non-broadcasters,
  or $\m{ready}{v}$ from $2f{+}1$ parties:
  deliver $v$.\;
\Indm
\end{algorithm}

\begin{figure}[t]
\centering
\begin{tikzpicture}[>=Latex, font=\scriptsize,
  m/.style={draw, rounded corners=2pt, minimum height=4.5mm,
            minimum width=9.5mm, inner sep=2pt},
  d/.style={align=center, text=black!75}]
\node[m] (p) at (0,0) {\msg{propose}};
\node[m] (e) at (1.7,0) {\msg{echo}};
\node[m] (c) at (3.4,0) {\msg{confirm}};
\node[m] (v) at (3.4,-1.1) {\msg{vote}};
\node[m] (a) at (5.1,0) {\msg{ack}};
\node[m] (r) at (6.8,0) {\msg{ready}};
\draw[->] (p) -- (e);
\draw[->] (e) -- node[above]{$C$} (c);
\draw[->] (e.south) |- node[pos=.75,below]{$M$} (v.west);
\draw[->] (c) -- node[above]{$f{+}1$} (a);
\draw[->] (v.east) -| node[pos=.25,below]{$C$} (a.south);
\draw[->] (a) -- node[above]{$C$} (r);
\draw[->] (r.south) to[out=-60,in=-120,looseness=5]
  node[below]{$f{+}1$} (r.south);
\node[d, above=1.5mm of e] {deliver on\\$Q$ echoes};
\node[d, above=1.5mm of c] {deliver on\\$C$ confirms};
\node[d, above=1.5mm of r] {deliver on\\$2f{+}1$ readies};
\end{tikzpicture}
\caption{Message flow of Algorithm~\ref{alg:rbc}.  Arrow labels give the
number of distinct senders that trigger the next message.  The \msg{vote},
\msg{ack}, \msg{ready} path is the recovery path of Shrestha et al.;
\msg{confirm} and delivery on $C$ confirms are new.  With an honest
broadcaster, every honest party has the echoes after two steps and the
confirms after three, and delivers on the echoes already after two steps
when at most $b$ parties are Byzantine.}
\Description{The broadcaster proposes, then non-broadcasters echo. Echoes
trigger votes and confirms; votes or confirms trigger acknowledgments;
an acknowledgment quorum or ready amplification triggers ready. Delivery
uses an echo quorum, a confirm quorum, or a ready quorum.}
\label{fig:flow}
\end{figure}
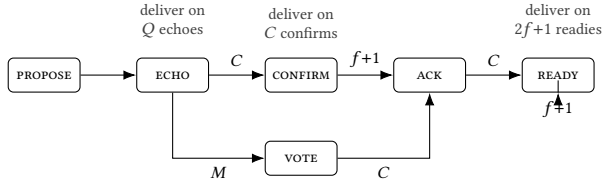

\begin{table}[t]
\caption{Roles of the threshold inequalities.  When the broadcaster is
Byzantine, any $f$ non-broadcasters include an honest party.}
\label{tab:thresholds}
\small
\begin{tabular}{@{}>{\raggedright\arraybackslash}p{2.3cm}>{\raggedright\arraybackslash}p{5.6cm}@{}}
\toprule
Condition & Role in the proof\\
\midrule
$f<M\le C\le N-f$ &
Votes and confirms rest on an honest echo; honest echoes and confirms
suffice in the good case.\\
\addlinespace
$2C\ge N+f$ &
Two confirm quorums, or two ack quorums, share an honest party; a delivery
on confirms excludes every conflicting vote quorum.\\
\addlinespace
$Q+M\ge N+f$ &
A delivery on echoes excludes conflicting honest votes and confirms.\\
\addlinespace
$Q-f+1\ge M$ &
A delivery on echoes leaves enough honest echoes to trigger votes
everywhere.\\
\addlinespace
$C+M>N$ &
Counting only the first echo from each party, no honest party can confirm
one value and vote for another.\\
\addlinespace
$C-f+1\ge f+1$ &
A delivery on confirms leaves enough honest confirms to trigger acks
everywhere.\\
\addlinespace
$Q\le N-b$ &
Honest echoes suffice for two-step delivery with at most $b$ Byzantine
parties.\\
\bottomrule
\end{tabular}
\end{table}

\paragraph{Why an ack is triggered by $f+1$ confirms}
Two natural alternatives for the confirm trigger of the \msg{ack} rule are
the $2f$ echoes used by Shrestha et al.\ and the requirement that the ack
sender itself voted for $v$.  Both break totality.  Let $n=10$ and $f=3$, so
$b=2$ and the thresholds are $5$ for \msg{vote}, $6$ for \msg{confirm},
\msg{ack}, and \msg{ready}, and $7$ for delivery on echoes.  The broadcaster
and two non-broadcasters are Byzantine.  The broadcaster proposes $A$ to four
honest parties and $B$ to the other three.  The Byzantine non-broadcasters
echo $A$ to the first group and $B$ to the second.  The first group receives
six echoes for $A$ and votes and confirms $A$.  The second group receives
five echoes for $B$ and votes for $B$, but has only four echoes for $A$ and
never confirms.  Deliver the four honest confirms for $A$ everywhere, and the
two Byzantine confirms for $A$ to one party of the first group, which
delivers $A$ on six confirms.  The Byzantine parties now send nothing.  Under
either alternative rule only the four members of the first group can ack
$A$, fewer than six, so no honest party ever sends \msg{ready}, and the six
other honest parties never deliver.  In Algorithm~\ref{alg:rbc} all seven
honest non-broadcasters ack $A$ upon the four honest confirms, including the
three that voted for $B$, and every honest party delivers within three
further steps.

\paragraph{Why the ack layer cannot be bypassed}
One might remove \msg{ack} and send \msg{ready} directly upon $2f$ votes or
$f+1$ confirms, keeping ready amplification and the three delivery rules.
This modification also fails at $n=10$, $f=3$.  Let the broadcaster and two
non-broadcasters be Byzantine, and split the honest echoes as four for $A$
and three for $B$.  Choose two of the four $A$ echoers.  To these two, the
Byzantine non-broadcasters echo $A$; delivering six echoes for $A$ makes
both vote and confirm $A$.  To the other five honest parties, the Byzantine
non-broadcasters echo $B$; delivering the three honest and two Byzantine
echoes for $B$ makes all five vote $B$.  Each receiver counts only the
first echo from each sender, so delivering every remaining honest echo
changes none of these choices.

There are now two honest confirms for $A$ and five honest votes for $B$.
The two Byzantine non-broadcasters send confirms for $A$, and one sends a
vote for $B$.  Selective scheduling gives four honest parties the four
confirms for $A$ first, and the other three the six votes for $B$ first.
In the modified protocol they respectively send \msg{ready} for $A$ and
$B$.  The three Byzantine parties send ready for $A$ only to one of the
four $A$-ready parties, which delivers on seven readies.  They send nothing
further.  After all honest messages are delivered, the other parties have
only four readies for $A$ and three for $B$, and each has already sent its
one ready.  No echo or confirm delivery quorum exists there.  Totality
fails.  The ack layer permits these conflicting preliminary choices but
requires a unique quorum of six acks before any ready is seeded.  This
example concerns that particular shortcut, not a lower bound on the number
of message types of every reliable broadcast protocol.

\subsection{Safety}
\label{sec:safety}

Throughout this subsection, the thresholds satisfy the shared hypotheses
of Section~\ref{sec:protocol}.  If the broadcaster is Byzantine, at most
$f-1$ of the $N$ non-broadcasters are Byzantine.  Thus the inequalities
$2C\ge N+f$, $Q+M\ge N+f$, and $2Q\ge N+f$ make the corresponding
quorums intersect in an honest party.  Honest means never corrupted;
Section~\ref{sec:model} permits adaptive choice of the corrupted set.

\begin{lemma}[Honest broadcaster]
\label{lem:honest-broadcaster}
If the broadcaster is honest and invoked with $v$, every honest \msg{echo},
\msg{vote}, \msg{confirm}, \msg{ack}, and \msg{ready} message is for $v$, no
honest party delivers before the invocation, and every honest party that
delivers delivers $v$.
\end{lemma}
\begin{proof}
An honest echo answers the broadcaster's proposal, so it is for $v$ and
follows the invocation.  Consider the honest messages in the order in which
they are sent.  A vote needs $M>f$ echoes and a confirm needs $C>f$
echoes, so each includes an honest echo, which is for $v$.  An ack needs
$C>f$ votes or $f+1$ confirms, so it includes an honest vote or confirm.  A
ready needs $C$ acks, including an honest ack, or $f+1$ readies, including
an earlier honest ready.  By induction all these messages are for $v$ and
follow the invocation.  Each delivery rule needs more than $f$ messages of
one type, hence an honest one, so every honest delivery is of $v$ and follows
the invocation.
\end{proof}

For the rest of this section the broadcaster is Byzantine.  The lemmas
concern the entire execution, not only the events after a delivery.

\begin{lemma}[A confirm excludes a vote]
\label{lem:confirm-vote}
An honest party that sends $\m{confirm}{v}$ never sends $\m{vote}{w}$ for
$w\ne v$, whichever of the two messages it sends first.
\end{lemma}
\begin{proof}
It holds $\m{echo}{v}$ from $C$ non-broadcasters and counts only the
first echo from each party.  A vote for $w$ would need $\m{echo}{w}$ from
$M$ further non-broadcasters, and $C+M>N$.
\end{proof}

\begin{lemma}[Confirms agree; delivery on echoes excludes conflicting votes]
\label{lem:confirm-agree}
All honest confirms are for the same value.  If some honest party has
received $\m{echo}{v}$ from $Q$ non-broadcasters, then no honest party
sends $\m{vote}{w}$ or $\m{confirm}{w}$ for $w\ne v$.
\end{lemma}
\begin{proof}
Two honest confirms rest on two sets of $C$ echo senders, which share at least $f$
non-broadcasters by~\eqref{eq:general-inequalities}, including an honest one.  That
party sends one echo, so the two confirms are for the same value.  An echo
quorum of $Q$ for $v$ and a vote quorum of $M$ for $w$ likewise share at least
$f$ non-broadcasters, including an honest one, so $v=w$.  A confirm for $w$
rests on $C\ge M$ echoes for $w$, so the same argument applies to it.
\end{proof}

\begin{lemma}[Early delivery excludes conflicting honest acks]
\label{lem:delivery-ack}
Suppose some honest party has received $\m{echo}{v}$ from $Q$
non-broadcasters or $\m{confirm}{v}$ from $C$ non-broadcasters.  Then every
honest ack in the execution, whether sent before or after that event, is for
$v$.
\end{lemma}
\begin{proof}
In the echo case, Lemma~\ref{lem:confirm-agree} excludes honest votes and
confirms for other values.  Both triggers of the \msg{ack} rule contain an
honest vote or confirm, so no honest party acks another value.

In the confirm case, at least $C-(f-1)$ of the confirm senders are
honest.  By Lemma~\ref{lem:confirm-vote} none of them ever votes for
$w\ne v$, so a vote quorum for $w$ would have to consist of $C$ of the
remaining $N-(C-f+1)$ non-broadcasters.  This is impossible because
$N-(C-f+1)<C$ follows from $2C\ge N+f$.  A set of
$f+1$ confirms for $w$ would contain an honest confirm for $w$, contradicting
Lemma~\ref{lem:confirm-agree}.  Hence neither trigger of the \msg{ack} rule
fires for $w$ at an honest party, at any time.
\end{proof}

\begin{lemma}[Ack quorums and readies agree]
\label{lem:ready}
If honest parties receive $\m{ack}{v}$ from $C$ non-broadcasters and
$\m{ack}{w}$ from $C$ non-broadcasters, then $v=w$.  Every honest ready is
for this value, and every honest ready is preceded, in causal order, by an
honest ready sent upon $C$ acks.
\end{lemma}
\begin{proof}
The two ack quorums share at least $f$ non-broadcasters, including an honest one, who
sends a single ack.  Now consider an honest ready.  If it was sent upon $f+1$
readies, one of those is an honest ready sent earlier.  Following such
predecessors, which are strictly earlier in the execution, we reach an honest
ready sent upon $C$ acks, and all readies along the way carry its value.
Every such starting ready carries the unique ack-quorum value.
\end{proof}

\begin{theorem}[Agreement and validity of deliveries]
\label{thm:safety}
Algorithm~\ref{alg:rbc} satisfies agreement, and each honest party
delivers at most once.  If the broadcaster is honest, every honest delivery
is of its input value and follows its invocation.  Eventual delivery is
proved in Section~\ref{sec:liveness}.
\end{theorem}
\begin{proof}
The delivery rule fires at most once per party.
Lemma~\ref{lem:honest-broadcaster} gives validity of deliveries and agreement with an
honest broadcaster.  Let the broadcaster be Byzantine and consider two honest
deliveries, of $v$ and $w$.  Two echo quorums share an honest party
by~\eqref{eq:general-inequalities}, so two deliveries on echoes agree.  A delivery
on confirms includes an honest confirm, so two deliveries on confirms agree
by Lemma~\ref{lem:confirm-agree}, and a delivery on echoes and one on
confirms agree by the same lemma.  Two deliveries on readies include honest
readies and agree by Lemma~\ref{lem:ready}.  Finally, let $v$ be delivered on
echoes or confirms and $w$ on readies.  By Lemma~\ref{lem:delivery-ack} every
honest ack is for $v$; an ack quorum contains an honest ack, so by
Lemma~\ref{lem:ready} every honest ready, including those in the ready
quorum for $w$, is for $v$.  Hence $v=w$.
\end{proof}

\subsection{Liveness, latency, and communication}
\label{sec:liveness}

Throughout this subsection the thresholds satisfy~\eqref{eq:general-inequalities}.
The number $t$ of Byzantine parties counts every party corrupted during the
execution.

\begin{lemma}[Good case]
\label{lem:good-case}
If the broadcaster is honest and invoked at time $S$, every honest party
delivers by $S+3\delta$.  If moreover at most $b$ parties are Byzantine,
every honest party delivers by $S+2\delta$.
\end{lemma}
\begin{proof}
Every honest non-broadcaster receives the proposal and echoes by $S+\delta$.
There are at least $N-f\ge C$ honest non-broadcasters, and their echoes reach
every party by $S+2\delta$.  By Lemma~\ref{lem:honest-broadcaster} no honest
party has voted or confirmed another value, so every honest non-broadcaster
confirms $v$ by $S+2\delta$.  These at least $C$ confirms reach every party,
including the broadcaster, by $S+3\delta$, and the delivery rule on $C$
confirms fires.  If at most $b$ parties are Byzantine, at least $N-b\ge Q$
honest echoes reach every party by $S+2\delta$, and the delivery rule on
echoes fires.
\end{proof}

\begin{table}[t]
\caption{Upper bounds after the first honest delivery at $T$, with a
Byzantine broadcaster.  ``$+k$'' means by $T+k\delta$. All entries
concern full-value delivery in the confirm/ack protocol.}
\label{tab:totality}
\small\setlength{\tabcolsep}{3pt}
\begin{tabular}{@{}llcc@{}}
\toprule
First delivery & Honest messages sent by $T$ & General & $t=1$\\
\midrule
$Q$ echoes & $Q-f+1\ge M$ echoes & $+4$ & $+1$\\
$C$ confirms & $C-f+1\ge f+1$ confirms & $+3$ & $+1$\\
$2f+1$ readies & $f+1$ readies & $+2$ & $+2$\\
\bottomrule
\end{tabular}
\end{table}

\begin{lemma}[Totality delay]
\label{lem:totality}
If the first honest delivery happens at time $T$, every honest party
delivers by $T+4\delta$.  Table~\ref{tab:totality} gives the bound for each
delivery rule.
\end{lemma}
\begin{proof}
With an honest broadcaster, $T\ge S$ by Lemma~\ref{lem:honest-broadcaster}
and every honest party delivers by $S+3\delta\le T+3\delta$ by
Lemma~\ref{lem:good-case}.  Let the broadcaster be Byzantine and let $v$ be
the value delivered at $T$.

\emph{Delivery on echoes.}  At least $Q-f+1\ge M$ honest parties have
echoed $v$ by $T$, and these echoes reach every honest party by $T+\delta$.
By Lemma~\ref{lem:confirm-agree} no honest party has voted for another
value, so all honest non-broadcasters, at least $n-f$ of them, vote for $v$
by $T+\delta$.  Their votes reach everyone by $T+2\delta$; by
Lemma~\ref{lem:delivery-ack} no honest ack went to another value, so all
honest non-broadcasters ack $v$ by $T+2\delta$.  These at least $n-f>C$
acks reach everyone by $T+3\delta$; by Lemma~\ref{lem:ready} no honest ready
went to another value, so every honest party sends $\m{ready}{v}$ by
$T+3\delta$, and at least $n-f\ge2f+1$ honest readies reach every honest
party by $T+4\delta$.

\emph{Delivery on confirms.}  At least $C-f+1\ge f+1$ honest parties have
confirmed $v$ by $T$.  Their confirms reach every honest non-broadcaster by
$T+\delta$ and trigger $\m{ack}{v}$ there, whatever that party voted for;
Lemma~\ref{lem:delivery-ack} excludes an earlier ack for another value.  All
honest parties therefore send \msg{ready} by $T+2\delta$ and deliver by
$T+3\delta$.

\emph{Delivery on readies.}  At least $f+1$ honest parties have sent
$\m{ready}{v}$ by $T$.  These readies reach every honest party by $T+\delta$
and trigger its ready if it has not sent one; by Lemma~\ref{lem:ready} no
honest ready went to another value.  At least $n-f\ge2f+1$ honest readies
reach everyone by $T+2\delta$.

By Theorem~\ref{thm:safety} all these deliveries are of $v$.  If the first
delivery satisfies several rules, each applicable bound holds.
\end{proof}

\begin{theorem}[Full-value protocol]
\label{thm:bounded}
Under~\eqref{eq:general-inequalities}, Algorithm~\ref{alg:rbc} is an
asynchronous reliable broadcast protocol against adaptive corruption of at
most $f$ parties, and it uses no cryptography beyond authenticated
channels.  Its optimistic good-case latency is at most two steps with
threshold $b$, its good-case latency at most three, and its totality delay
at most four.  The messages sent by parties while they follow the protocol
number at most $(5(n-1)+2)(n-1)$ and carry $O(n^2(B+\log n))$ bits in every
execution, including messages after delivery and executions without
delivery.
\end{theorem}
\begin{proof}
Agreement and the correctness of delivered values are
Theorem~\ref{thm:safety}.  The arguments of Lemmas~\ref{lem:good-case}
and~\ref{lem:totality} are chains of finitely many message transmissions
between honest parties, each of which completes under eventual receipt and
fair processing; they therefore prove eventual delivery for an honest
broadcaster and totality for any broadcaster, and under the timing
assumption of Definition~\ref{def:latency} they give the stated bounds.  For
correctness and latency, fix the set of at most $f$ parties ever corrupted;
their entire transcripts are admissible Byzantine behavior, and the other
parties obey the rules throughout, so every quorum argument above applies
unchanged.  For the message bound, each non-broadcaster sends at most one
message of each of five types to each of the $n-1$ other parties while it
follows the protocol, and the broadcaster sends at most its proposal and
one ready.  A message has $B+O(\log n)$ bits: the value, the type, and the
party and instance identifiers.  Since each type is sent at most once, these
bounds hold whether or not parties deliver.
\end{proof}

\begin{corollary}[Maximal optimistic threshold]
\label{cor:general-optimal}
\label{cor:optimal}
For $3f+1\le n\le4f-1$, the thresholds of Lemma~\ref{lem:admissible} with
$b=\lfloor n/2\rfloor-f$ give the optimal good-case latency three and two-step
optimistic good-case latency under the largest $b$ that any signature-free
protocol admits.  For $n\ge4f$, $b=f$ is admissible and gives good-case
latency two under all $f$ faults.  Theorem~\ref{thm:main} follows from
Theorem~\ref{thm:bounded} with these thresholds.
\end{corollary}
\begin{proof}
In the first range $b<f$, so it is the largest value permitted by
Lemma~\ref{lem:admissible}; there $M=\lceil n/2\rceil$.  The good-case lower
bound of Abraham, Ren, and Xiang requires $n\ge4f$ for two-step delivery
under all $f$ faults~\cite[Theorem~4]{AbrahamRenXiang2021Latency}, and
Theorem~8 of Shrestha et al.\ excludes optimistic thresholds above
$\lfloor(n-2f)/2\rfloor=\lfloor n/2\rfloor-f$ in this
range~\cite{ShresthaEtAl2026OptimisticRBC}.  Neither bound concerns the
totality delay or communication.  If $n\ge4f$, then
$\min\{f,\lfloor n/2\rfloor-f\}=f$, and the two-step guarantee of
Theorem~\ref{thm:bounded} covers every fault pattern.
\end{proof}

\begin{remark}
At $n=3f+1$, $f=1$, the maximal threshold is $b=1$, and every threshold on
non-broadcasters equals $2$, so a party may deliver on echoes at the moment
it votes and confirms.  Continued processing after delivery is needed for
totality already in this smallest case.
\end{remark}

\subsection{Exact fault-dependent totality}
\label{app:totality-tightness}

We first refine the upper bound in terms of the actual number of faults,
then give matching executions at minimum resilience.
Throughout this subsection the thresholds satisfy~\eqref{eq:general-inequalities}.

\begin{lemma}[Fault-dependent totality]
\label{lem:fault-sensitive}
Assume $b\ge1$, and let $t\le f$ be the number of parties ever corrupted.
If the first honest delivery is at $T$, every honest party delivers by
$T+U(t)\delta$, where
\[
 U(t)=
 \begin{cases}
 1,&t=0,\\
 2,&t=1,\\
 3,&2\le t\le Q-C+1,\\
 4,&\max\{1,Q-C+1\}<t\le f.
 \end{cases}
\]
Empty ranges are omitted.  These are upper bounds for this protocol.  If the
broadcaster is Byzantine and $t=1$, the echo and confirm branches each take
at most \emph{one} further step; only the ready branch uses the two-step
bound.
\end{lemma}
\begin{proof}
If $t=0$, every message in the first delivery's echo, confirm or ready
quorum was sent by an honest party by $T$.  The same quorum reaches
every party by $T+\delta$, so all deliver on that rule.

An honest broadcaster gives the absolute deadline $S+2\delta$ when $t\le b$,
and $S+3\delta$ otherwise.  Since $T\ge S$ and $b\ge1$, these imply the
claimed relative bounds for $t\ge1$.

Suppose the broadcaster is Byzantine, so $t\ge1$ and at most $t-1$
non-broadcasters are Byzantine.  A first delivery on $Q$ echoes contains at
least $Q-t+1$ honest echoes, all received everywhere by $T+\delta$.  If
$t\le Q-C+1$, there are at least $C$ of them: every honest non-broadcaster
confirms by $T+\delta$, since by Lemma~\ref{lem:confirm-agree} none has
confirmed another value, and everyone delivers on confirms by $T+2\delta$.
Otherwise Lemma~\ref{lem:totality} gives four steps.  A first delivery on
confirms always has the three-step bound, and one on readies the two-step
bound.

When $t=1$, every non-broadcaster is honest.  All $Q$ echoes, or all $C$
confirms, behind a delivery on echoes or confirms were therefore sent to
everyone and arrive by $T+\delta$, so everyone delivers on the same rule
within one step.  The ready branch keeps the bound two, which gives
$U(1)=2$.

\end{proof}

\begin{proposition}[Exact confirm/ack totality profile]
\label{prop:four-tight}
\label{prop:confirm-ack-profile}
For Algorithm~\ref{alg:rbc} with the thresholds of
Lemma~\ref{lem:admissible}, at $n=3f+1$, $f\ge1$ and any
$1\le b\le\lfloor(f+1)/2\rfloor$, the supremum of the totality delay over
executions with exactly $t$ Byzantine parties is
\[
 D_{f,b}(t)=
 \begin{cases}
 1,&t=0,\\
 2,&t=1,\\
 3,&2\le t\le f-b+1,\\
 4,&f-b+2\le t\le f.
 \end{cases}
\]
Empty ranges are omitted. This statement does not assert optimality
among reliable broadcast protocols.
\end{proposition}
\begin{proof}
Lemma~\ref{lem:fault-sensitive} gives the upper bounds.  We give matching
executions with $\delta=1$, writing $C=2f$, $Q=3f-b$ and $M=2f+1-b$.
Except for the transmissions specified below, every honest message
arrives after one time unit, including messages sent after delivery;
self-messages are immediate.  Byzantine parties send only the specified
messages.  Zero-delay transmissions can instead be arbitrarily short
for the supremum claim.

\emph{No faults.}  Invoke the broadcaster at zero.  Give its proposal to
$Q$ non-broadcasters immediately and to the others at time one.  Choose
$p$ among the first group and deliver their echoes to $p$ immediately.
Party $p$ delivers at zero, whereas the broadcaster receives its first
echoes and delivers at one.

\emph{One fault.}  Let one non-broadcaster be Byzantine.  Give the honest
proposal immediately to $Q-1$ honest non-broadcasters and at time one
to the other $b$ honest non-broadcasters.  Deliver the first group's
echoes and the Byzantine party's matching echo immediately to one of
them, $p$.  It delivers at zero.  The broadcaster gets only $Q-1$ echoes
at time one and receives the remaining honest echoes at time two.
Before time two it has only $p$'s confirm, insufficient for delivery;
no earlier ready quorum exists.  It therefore delivers at two.

\emph{Three delays.}  Let $2\le t\le f$, including the broadcaster.
Give a proposal only to a set $I$ of $\ell=2f-t+1$ honest parties, at zero.
Deliver their echoes to every member of $I$ immediately, together with
the $t-1$ Byzantine non-broadcasters' matching echoes.  Every member
of $I$ has $C$ echoes, votes and confirms at zero.  Choose $p\in I$;
deliver the $\ell$ honest confirms and $t-1$ Byzantine confirms to $p$
immediately; no other honest party receives Byzantine confirms.
Party $p$ delivers and acks at zero.  No fast delivery is
possible, since only $C<Q$ matching echoes exist.  At time one every honest
party has the $\ell\ge f+1$ honest confirms and has sent its ack; these acks cause readies
at two, and the readies cause delivery at three.  Except at $p$, the
$\ell<C$ confirms never complete a direct-delivery quorum.  Only $p$
acks before one, so no earlier ready quorum occurs.

\emph{Four delays.}  Let $f-b+2\le t\le f$ and fix $0<\varepsilon<1$.
The broadcaster and $t-1$ non-broadcasters are Byzantine.  Put
$e=Q-t+1$, so $M\le e<C$.  Give the proposal to one honest party $p$
at zero and to $e-1$ others at $1-\varepsilon$.  Deliver all their echoes
to $p$ by $T=1$, together with the $t-1$ Byzantine echoes.  It has $Q$
echoes and delivers at $T$, voting and confirming then.
Everyone else completes the $e$ honest echoes at $2-\varepsilon$ and
votes then, but never confirms: $e<C$.  The single confirm from $p$
cannot trigger an ack.  Before the other votes arrive, each such voter
has at most its own vote and $p$'s, fewer than $C=2f$; the nonempty
four-delay range requires $f\ge3$.  The other votes complete the $C$-vote quorum
at $3-\varepsilon$, causing acks, then readies at $4-\varepsilon$, and
delivery at $5-\varepsilon$.  The gap is $4-\varepsilon$.

Every honest message has delay at most one in these executions.
\end{proof}

At maximal optimism, $b=b_*=\lfloor(f+1)/2\rfloor$, four steps occur only
above $f-b_*+1$ faults. At $n=7,f=2,b=1$, the four-step range is empty.
These matching executions establish tightness for our confirm/ack
construction, not a lower bound for reliable broadcast in general.

\end{document}